\documentclass[sn-nature]{sn-jnl}%

\usepackage{graphicx}%
\usepackage{multirow}%
\usepackage{amsmath,amssymb,amsfonts}%
\usepackage{amsthm}%
\usepackage{mathrsfs}%
\usepackage[title]{appendix}%
\usepackage{xcolor}%
\usepackage{textcomp}%
\usepackage{manyfoot}%
\usepackage{booktabs}%
\usepackage{algorithm}%
\usepackage{algorithmicx}%
\usepackage{algpseudocode}%
\usepackage{listings}%

\usepackage[utf8]{inputenc}
\usepackage{mathtools}
\usepackage{graphicx}
\usepackage{textcomp}
\usepackage{xcolor}
\usepackage{physics}
\usepackage{url}
\usepackage{bbm}
\usepackage{bm}
\usepackage{nicefrac}
\usepackage{appendix}
\usepackage{hyperref}
\usepackage{braket}
\usepackage{layouts}
\usepackage{pifont}
\usepackage[export]{adjustbox}
\usepackage[caption=false,position=top,labelfont=bf,textfont=normalfont,singlelinecheck=off,justification=raggedright]{subfig}

\theoremstyle{thmstyleone}%
\newtheorem{theorem}{Theorem}%

\theoremstyle{thmstyletwo}%

\theoremstyle{thmstylethree}%

\newtheorem{manualcorollaryinner}{Corollary}[theorem]
\newenvironment{manualcorollary}[1]{%
  \renewcommand\themanualcorollaryinner{#1}%
  \manualcorollaryinner
}{\endmanualcorollaryinner}

\newtheorem{manuallemmainner}{Lemma}[theorem]
\newenvironment{manuallemma}[1]{%
  \renewcommand\themanuallemmainner{#1}%
  \manuallemmainner
}{\endmanuallemmainner}

\newtheorem{manualremarkinner}{Remark}[theorem]
\newenvironment{manualremark}[1]{%
  \renewcommand\themanualremarkinner{#1}%
  \manualremarkinner
}{\endmanualremarkinner}

\newcommand{\bx}{\bm{x}}
\newcommand{\F}{\mathcal{F}}
\newcommand{\G}{\mathcal{G}}
\newcommand{\U}{\mathcal{U}}
\newcommand{\calH}{\mathcal{H}}
\newcommand{\calP}{\mathcal{P}}

\newcommand{\Gate}[1]{\textsc{#1}}

\newcommand{\zgate}{\Gate{z}}
\newcommand{\ygate}{\Gate{y}}
\newcommand{\rygate}{\Gate{Ry}}
\newcommand{\xgate}{\Gate{x}}

\newcommand{\tgate}{\Gate{t}}

\newcommand{\idgate}{\Gate{I}}

\newcommand{\cnotgate}{\Gate{cnot}}

\makeatletter
\newcommand{\vast}{\bBigg@{4}}
\newcommand{\Vast}{\bBigg@{5}}
\makeatother

\makeatletter
\newcommand*{\centerfloat}{%
  \parindent \z@
  \leftskip \z@ \@plus 1fil \@minus \textwidth
  \rightskip\leftskip
  \parfillskip \z@skip}
\makeatother

\def \arxiv{z}  %
\newcommand{\alt}[2]{\ifx\arxiv\undefined{#1}\else{#2}\fi}  %

\alt{\unnumbered}{}  %

\begin{document}

\title[Constrained Optimization via Quantum Zeno Dynamics]{Constrained Optimization via Quantum Zeno Dynamics}

\author*[1]{\fnm{Dylan} \sur{Herman}}\email{dylan.a.herman@jpmchase.com}
\equalcont{These authors contributed equally to this work.}

\author[1]{\fnm{Ruslan} \sur{Shaydulin}}
\equalcont{These authors contributed equally to this work.}

\author[1]{\fnm{Yue} \sur{Sun}}
\equalcont{These authors contributed equally to this work.}

\author[1]{\fnm{Shouvanik} \sur{Chakrabarti}}
\author[1]{\fnm{Shaohan} \sur{Hu}}
\author[1]{\fnm{Pierre} \sur{Minssen}}
\author[1]{\fnm{Arthur} \sur{Rattew}}
\author[1]{\fnm{Romina} \sur{Yalovetzky}}
\author[1]{and~\fnm{Marco} \sur{Pistoia}}

\affil[1]{\orgdiv{Global Technology Applied Research}, \orgname{JPMorgan Chase}, \orgaddress{\city{New York}, \state{NY}, \postcode{10017}, \country{USA}}}

\abstract{
Constrained optimization problems are ubiquitous in science and industry. Quantum algorithms have shown promise in solving optimization problems, yet none of the current algorithms can effectively handle arbitrary constraints.
We introduce a technique that uses quantum Zeno dynamics to solve optimization problems with multiple arbitrary constraints, including inequalities. We show that the dynamics of quantum optimization can be efficiently restricted to the in-constraint subspace on a fault-tolerant quantum computer via repeated projective measurements, requiring only a small number of auxiliary qubits and no post-selection. 
Our technique has broad applicability, which we demonstrate by incorporating it into the quantum approximate optimization algorithm (QAOA) and variational quantum circuits for optimization. We evaluate our method numerically on portfolio optimization problems with multiple realistic constraints and observe better solution quality and higher in-constraint probability than state-of-the-art techniques. We implement a proof-of-concept demonstration of our method on the Quantinuum H1-2 quantum processor. }

\maketitle

\section{Introduction}
Combinatorial optimization is widely considered to be one of the most promising problem domains for quantum algorithms. The ubiquity of hard optimization problems in science and industry amplifies the impact of any improvements in algorithmic performance. In practice, the optimization problems often have many constraints, such as the regulatory constraints when optimizing a portfolio or logistic constraints when optimizing flight crew assignments. Being able to incorporate a diverse range of constraints is an essential criterion for the applicability of a quantum algorithm to industrial problems.

A commonly considered class of quantum optimization algorithms uses a parameterized quantum evolution to drive the quantum system towards a state encoding the solution of the optimization problem. This class of algorithms includes the quantum approximate optimization algorithm (QAOA)~\cite{Hogg2000,farhi2014quantum} and variational algorithms for optimization~\cite{cerezo2021variational, rattew2019}. {While these algorithms are often discussed as promising approaches for noisy near-term devices~\cite{kandala2017hardware,2303.02064}, many results supporting their potential are analytically derived or numerically demonstrated in the fault-tolerant regime~\cite{2205.12481,2208.06909,qaoa-labs}. Therefore, in this paper we primarily view these algorithms as targeting fault-tolerant quantum processors.}

One of the main challenges in applying these quantum algorithms to commercially-relevant optimization problems is the need to enforce the constraints. Concretely, the goal is to prepare a quantum state such that upon measuring it, a high-quality solution that satisfies the constraints is obtained with high probability. Two commonly considered approaches are to encode the constraint into the objective using a penalty term and to directly restrict the parameterized quantum evolution to the in-constraint subspace. In the first approach, a penalty term is added to the objective so that optimizing the objective requires satisfying the constraint. While such approaches are flexible enough to satisfy most constraints, the quality of the result is sensitive to the choice of the penalty strength~\cite{Wang2020}. As tuning the penalty strength is difficult in general, this approach often leads to sub-optimal performance in practice~\cite{Niroula2022}. This observation motivates the second approach, i.e., restricting the quantum evolution to the in-constraint subspace.

A number of techniques have been proposed to ensure that the parameterized quantum evolution respects the constraints of the problem. Hadfield et al.~\cite{hadfield2018quantum, Hadfield_2019} proposed the quantum alternating operator ansatz algorithm, which applies pairs of alternating operators to an in-constraint initial state. The first alternating operator (\emph{phase operator}) is diagonal in the computational basis and encodes the objective, and the second operator (\emph{mixing operator} or \emph{mixer}) is non-diagonal and restricts the transitions of probability amplitudes to the computational basis states corresponding to the in-constraint solutions. The problem of constructing a Hamiltonian preserving arbitrary constraints is $\mathsf{NP}$-complete even for linear constraints~\cite{Leipold2021}, though explicit constructions are available for some combinatorial optimization problems~\cite{hadfield2018quantum, Hadfield_2019,Stollenwerk2020,Hen2016}. In general, constraint-preserving mixers are difficult to implement, even when constructions are available~\cite{Cook2020,fuchs2022constrained}. The cost of implementing the algorithm on hardware can be reduced for a restricted class of problems by combining the phase and mixing operators~\cite{larose2021mixer}. If a uniform superposition of in-constraint states can be prepared efficiently, a Grover operator can be used as the mixer~\cite{bartschi2020grover, Gilliam_2021,Golden2021}. Finally, for problems with an indexable set of feasible states (such as those with Hamming-weight constraints), a continuous-time quantum walk in the solution space can be used as a mixer~\cite{Marsh2019AQW, Marsh_2020, Slate2021quantumwalkbased}. However, none of these techniques are sufficiently flexible to handle the general case of multiple arbitrary constraints directly. { The parity optimization framework \cite{paritycompiler_2021, parityconstraints_2021, paritybenchmarks_2021, modularparity_2022, Dominguez2023} can natively handle polynomial equality constraints for QAOA-like circuits. However, this framework introduces an auxiliary qubit for every unique monomial term that appears, leading to large space overhead for complex objectives and constraints.} All of the techniques mentioned above consider QAOA-like alternating operator circuits, and are not easy to generalize to other variational algorithms. 

In this work, we introduce an approach for enforcing multiple arbitrary constraints in quantum optimization. We restrict the quantum evolution to the in-constraint subspace by repeated projective measurements. In each measurement, the value of the constraint is computed onto an auxiliary register, which is then measured. Our technique uses quantum Zeno dynamics, wherein the evolution of the system is restricted to the subspace defined by the repeated projective measurements and transitions outside of this subspace are suppressed. Our approach is applicable to any problem in $\mathsf{NPO}$ (the $\mathsf{NP}$ optimization complexity class), as the only restriction we impose on the constraints is the existence of an efficient oracle for testing them. We provide explicit constructions for arbitrary combinatorial constraints. 
We demonstrate the effectiveness of the proposed technique by using it to enforce constraints in QAOA with various, unconstrained, mixing operators and the layer variational quantum eigensolver (L-VQE)~\cite{Liu2022}, which is a variational quantum algorithm for optimization. We show analytically that our technique is guaranteed to obtain the optimal in-constraint solution when applied to the digital simulation of the quantum adiabatic algorithm, or equivalently to QAOA in the constrained subspace with sufficiently large depth. We derive an analytical form of the scaling of the number of measurements required to maintain a constant minimum success probability for any {parameterized quantum evolution}. 
Furthermore, we provide numerical evidence that our technique, applied to QAOA for the portfolio optimization problem with a budget constraint, provides significant performance improvements over the state-of-the-art method of enforcing the constraint by introducing a penalty term. 
While the results we derive are for fault-tolerant quantum processors, high-fidelity near-term devices may be able to implement the algorithms without realizing full error-correction. To demonstrate an end-to-end realization of our technique, we implement QAOA with Zeno dynamics on the Quantinuum H1-2 trapped-ion quantum processor {for proof-of-concept portfolio optimization problems. These experiments complement our numerical simulations by using explicit constructions and compilations of circuits, including those for checking the constraints. In the hardware experiments, we} observe performance improvements from increasing the number of measurements, up to a two-qubit circuit depth of 148.

\section{Results}

\subsection{Quantum Zeno dynamics for constrained optimization}%

We now introduce our approach to enforcing constraints in quantum optimization by repeated non-selective projective measurements. Our method is general, though here we focus on algorithms utilizing parameterized states of the form 
\begin{equation}
    \label{eq:results_ansatz_vqe}
    \ket{\psi(\boldsymbol{\theta})} = U(\boldsymbol{\theta})\ket{s} = \prod_{j=1}^{m} e^{-i\theta_jH_j} \ket{s},
\end{equation}
where $H_j$ is some Hamiltonian, e.g., a tensor product of single-qubit Pauli operators, and $\ket{s}$ is the initial state, which lies in the system Hilbert space $\mathcal{H}$. 

A constrained combinatorial optimization problem has a set of feasible states $\mathcal{F}$, which is a subset of the $n$-dimensional Boolean cube $\mathbb{B}^n$.
Let $P_\F$ denote the orthogonal projector onto the subspace spanned by computational basis states corresponding to feasible solutions in $\F$. 
We discuss the construction of this operator in \alt{the \emph{Methods} Section}{Section~\ref{sec:implement_oracle}}. 
The measurement $\calP$ is a super-operator as defined as
\begin{equation}
    \label{eq:results_measurement}
    \calP \rho = \sum_{j=1}^k P_j \rho P_j,
\end{equation}
where $\sum_{j=1}^k P_j = \idgate$, and $P_j$ is a projection onto some subspace $\calH_{j} = P_j\calH$ of dimensionality $\Tr(P_j) \ge 1$.
Without loss of generality, we can assume $P_1 = P_\F$, and define $P_\G := \idgate - P_\F = \sum_{j=2}^k P_j$.

We give our main result in Theorem~\ref{thm:main_theorem}, which we use to derive the number of measurements required to enforce constraints {in parameterized evolutions of the form given by Equation~\eqref{eq:results_ansatz_vqe}.}

\begin{theorem}\label{thm:main_theorem}
Let $\calP$ be the measurement defined in Equation~\eqref{eq:results_measurement}. Suppose a system is evolved from some initial state $\rho_0 = P_j\rho_0P_j$ under the action of a Hamiltonian $H$, whose distinct eigenvalues are $\xi_{\min} =\xi_{1} < \xi_{2} < \cdots < \xi_{d}=\xi_{\max}$, for time $\theta$ . For {$\delta \leq 0.19$}, if $N$ applications of $\calP$ are performed at equally-spaced time intervals with
\begin{equation}
    \label{eqn:gen_meas_scaling}
    N = \left\lceil\frac{\left[\theta( \xi_{\max} - \xi_{\min})\right]^2}{\ln{\left({1-2\delta}\right)^{-2}}}\right\rceil,
\end{equation}
then the probability of measuring a state in $\mathcal{H}_j$ at time $\theta$ is lower bounded by {$1 - \delta$}, i.e.,
\begin{equation}
    \Tr\left[P_j\rho(\theta)\right] \geq  {1 - \delta},
\end{equation}
where
\begin{equation}
    \rho(\theta) = \U(\theta)\rho_0\U(\theta)^\dagger, \quad \U(\theta) = [\calP e^{-iH \, \theta/N}]^{N}.
\end{equation}
\end{theorem}
\begin{proof}
See \alt{the \emph{Methods} Section}{Section~\ref{sec:proof_of_main_thm}}.
\end{proof}
\begin{manualremark}{1}
\label{remark_1}
Note that since $2\lVert H \rVert_{2} \geq \lvert \xi_{\max} - \xi_{\min} \rvert$, the bound can be reformulated in terms of the spectral norm of the Hamiltonian. This may be useful as the spectral norm may be easier to bound in practice for complicated Hamiltonians.
\end{manualremark}

Assume that the initial state $\ket{s}$ respects the constraints, that is $P_\F\ket{s}=\ket{s}$. 
We apply a parameterized unitary $U(\boldsymbol{\theta})$ to the initial state following~Equation~\eqref{eq:results_ansatz_vqe}. 
{To enforce the constraints, we can insert measurements into the parameterized evolution as follows:}
{
\begin{align}
\label{eqn:ansatz_vqe_with_meas}
    \U_Z(\bm{\theta}) = \prod_{k=1}^{L}\left[\calP \prod_{j=1}^{m_{k}} e^{-i(\theta_{r(k,j)}/N_{k})H_{r(k,j)}} \ \right]^{N_{k}},
\end{align}
where $r(k, j) = \sum_{t=1}^{k-1}m_{t} + j$ and each sequence of $m_{k}$ parameterized evolutions, without a measurement, is called a \emph{block}. We define $N_{k} = 0$ to mean that no measurement is performed and  no $\theta_{r(k, j)}$ is not scaled for that block. The following corollarly provides a sufficient $N_{k}$ for each block to ensure a desired minimum in-constraint probability. The asymptotic dynamics, i.e. when $N_{k} \xrightarrow[]{} \infty, \;\forall k$ and also called the Zeno limit, will be different depending on how the blocks are chosen.
}
{
\begin{manualcorollary}{1}
\label{cor:scheme_2_cor}
Let $\calP$ be the measurement defined in Equation~\eqref{eq:results_ansatz_vqe}. Let the parameterized evolution defined in Equation~\eqref{eqn:ansatz_vqe_with_meas} evolve the system from some initial state $\rho_0 = P_j\rho_0P_j$. Then, in order to ensure that
$$\Tr[P_{j}\U_Z(\bm{\theta})\rho_{0}\U_Z(\bm{\theta})^{\dagger}] \geq 1 - \delta,$$
it suffices to choose
\begin{align}
N_k = \left\lceil\frac{4L[\sum_{j=1}^{m_k}\lvert\theta_{r(k,j)}\rvert]^2\max_{j}\lVert H_{r(k,j)}\rVert_{2}^{2}}{\tau(\delta)}\right\rceil,
\end{align}
where
\begin{itemize}
    \item $\tau(\delta) = \ln(1-2\delta)^{-2}$ if $H_{r(k,j)}$ pairwise commute,
    \item $\tau(\delta) =\ln{\left(1-\delta\right)^{-1.78}}$ otherwise,
\end{itemize}
and $\delta \leq 0.19$. In addition, the asymptotic dynamics is
\begin{align}
 \prod_{k=1}^{L}e^{-i \calP\bm{H}_{k}\cdot\bm{\theta}_{k}}\calP,
\end{align}
where $\calP$ acts element-wise on the vector $\bm{H}_{k} = (H_{(k,1)}, \dots , H_{(k,m_{k})})^{\mathsf{T}}$ and $\bm{\theta}_{k} = (\theta_{(k,1)}, \dots , \theta_{(k,m_{k})})$.
\end{manualcorollary}
}
\begin{proof}
See \alt{the \emph{Methods} Section}{Section~\ref{sec:proof_of_cor1}}.
\end{proof}
\begin{manualremark}{2}
For combinatorial optimization problems, constraint-preserving measurements that correspond to different constraints always commute. Thus $\mathcal{P}_{\mathcal{F}}$ can be implemented as a composition of measurements corresponding to different  constraints.
\end{manualremark}

While the previous results indicate that $N_k$ can grow inverse polynomially with the desired error probability, the following result (Corollary \ref{cor:efficient}) shows that fixing $\delta$ and applying a simple repetition scheme suffices to suppress the failure probability arbitrarily below $\delta$ with only logarithmic overhead. Thus, the overall procedure can be made efficient. The purpose of the Zeno framework is to ensure that we can obtain a state that has an overlap with $\mathcal{H}_j$ that is lower bounded by a constant and prepare this state with an overhead that is $O(\text{polylog}(\dim \mathcal{H}))$. 
{
\begin{manualcorollary}{2}
\label{cor:efficient}
Let $\calP$ be the measurement defined in Equation~\eqref{eq:results_measurement}. Let the parameterized evolution defined in Equation~\eqref{eqn:ansatz_vqe_with_meas} evolve the system from some initial state $\rho_0 = P_j\rho_0P_j$. In addition, suppose that the number of measurements $N_k$ was chosen, using Corollary \ref{cor:scheme_2_cor}, to ensure that 
$\Tr[P_{j}\rho_{Z}(\bm{\theta})] =\Tr[P_{j}\U_Z(\bm{\theta})\rho_{0}\U_Z(\bm{\theta})^{\dagger}]$
is lower bounded by a constant independent of the system size, and then in order to ensure that $\mathcal{P}$ applied to $\rho_{Z}(\bm{\theta})$ prepares a state in $\mathcal{H}_j$ with a probability at least $1 - \epsilon$, it suffices to
prepare and measure at most  $\log(1/\epsilon)$ copies of $\rho_{Z}(\bm{\theta})$.
\end{manualcorollary}
\begin{proof}
Suppose $\Tr[P_{j}\rho_{Z}(\bm{\theta})] = c$. Since we can efficiently check whether the post-measurement state obtained from applying $\mathcal{P}$ to $\rho_{Z}(\bm{\theta})$ is in $\mathcal{H}_j$, $\log(1/\epsilon)/\log(1/(1-c)) < \log(1/\epsilon)$ repetitions suffice to ensure that the outcome of at least one of the repetitions is in $\mathcal{H}_j$ with probability at least $1 - \epsilon$.
\end{proof}
}

{These results imply that for most practical cases, e.g. when $H_j$ are Pauli operators as in the cases of QAOA and hardware-efficient parameterized circuits, the number of measurements scales at most quadratically in the circuit depth and width, i.e., as $O(\text{polylog}(\dim \mathcal{H}))$. Thus, QZD can be used to efficiently constrain parameterized evolution for quantum optimization.}

\subsection{Constrained QAOA via Zeno dynamics}

We now discuss the application of QZD to QAOA. In a QAOA circuit, the phase operator $U_C(\gamma)$ is diagonal in the computational basis and cannot violate constraints. More specifically, it evolves the current state, for time $\gamma$, under the diagonal operator  $C=\sum_{\bm{x}\in\mathbb{B}^n}f(\bm{x})\ketbra{\bm{x}}$, which encodes the values of the objective function $f$ on $\mathbb{B}^n$.  The Hermitian mixing operator $B$ transitions probability amplitude between elements of  $\mathbb{B}^{n}$ and, in general, does not respect the problem constraints. Therefore the measurements only need to be added to the mixing operator. Since a $p$-layer QAOA circuit consists of $p$  applications of the phase and mixing operators in an alternating fashion, the full circuit combined with the Zeno framework then becomes
\begin{equation}\label{eq:qzd_qaoa_circuit}
    \U_{Z\text{-QAOA}}(\boldsymbol{\beta}, \boldsymbol{\gamma}) = \prod_{j=1}^{p}\left[\U_B(\beta_j, N_j)U_C(\gamma_j)\right],
\end{equation}
where
\begin{equation}
    \U_B(\beta_j, N_j) = \left[\calP e^{-i\frac{\beta_j}{N_j}B}\right]^{N_j}.
\end{equation}
In the notation of Equation \eqref{eqn:ansatz_vqe_with_meas}, this corresponds to setting all $m_{k} = 1$, and setting $N_k = 0$ for blocks containing the cost operator. While there are other valid choices for the blocks, the decomposition we have chosen is sufficient to achieve an efficient scheme.

As the mixing operator $B$ is known, we can explicitly derive  the number of measurements required to maintain a constant success probability. We observe that for any mixer this number of measurements grows linearly with the number of QAOA layers, and for commonly considered mixers, the number of measurements grows no more than quadratically with the number of qubits.

\begin{manualcorollary}{3}\label{cor:qaoa_bounds}
Let $\U_{Z\text{-QAOA}}(\boldsymbol{\beta}, \boldsymbol{\gamma})$ denote the QAOA circuit on $n$ qubits with $N$ measurements added to each mixing operator as defined in Equation~\eqref{eq:qzd_qaoa_circuit}. Let the initial state $\rho_0 = \ketbra{s}$ be in-constraint. Then $N_j$ measurements suffice to maintain at least a $1-\delta$ probability of obtaining an in-constraint measurement outcome, where
\begin{itemize}
    \item if $B=\sum_{k=1}^n\xgate_k$, then $N_j=\left\lceil\frac{p\beta_j^{2}n^2 }{\ln{\left[{1-2\delta}\right]^{-\frac{1}{2}}}}\right\rceil$ 
    \item if $B=\ketbra{+}$, then $N_j=\left\lceil\frac{p\beta_j^{2}}{\ln{\left[{1-2\delta}\right]^{-2}}}\right\rceil$,
\end{itemize}
and $\delta \leq 0.19$.
\end{manualcorollary}
\begin{proof}
The proof follows from Theorem \ref{thm:main_theorem} by noting that for $B=\sum_{k=1}^n\xgate_k$ the minimum and maximum eigenvalues are $-n$ and $n$, respectively, and for $B=\ketbra{+}$ the only eigenvalues are one and zero. For QAOA with $p$ layers, the number of measurements increases by a factor of $p$. Note that while we could  of instead used Corollary \ref{cor:scheme_2_cor}, using Theorem \ref{thm:main_theorem} directly results in $N_k$ being lower by a constant for $B=\ketbra{+}$.
\end{proof}
Note that the scaling rule of Corollary~\ref{cor:qaoa_bounds} implies that the number of measurements will change with $\beta_j$ and thus each mixer layer.

Figure~\ref{fig:meas_scale} visualizes how the number of measurements required to maintain a given minimum in-constraint probability, according to Corollary~\ref{cor:qaoa_bounds}, grows with the evolution time $\beta$ for the $B=\sum_j\xgate_j$ (\ding{54} marker) and $B=\ketbra{+}$ (\ding{58} marker) mixing operators for $p=1$ QAOA with a $3$-qubit initial state $\ket{s}$. As the phase operator is diagonal, there is no dependency on it. We note that the number of measurements for the mixer $B=\sum_j\xgate_j$ grows with number of qubits and is therefore larger than for $B=\ketbra{+}$. Note that when following the scaling rules of Corollary~\ref{cor:qaoa_bounds}, the number of measurements is multiplied by the number of QAOA layers $p$.

In \alt{the \emph{Results} Section}{Section~\ref{sec:numerical_experiments}}, we observe that for realistic constraints, the number of measurements is significantly lower. This is because the worst-case $P_\F$ and $\ket{s}$, i.e., from Equation~\eqref{eqn:worst_case_state} in the proof of Lemma \ref{lem:scaling_lemma}, are far from those encountered in practice. Specifically, the worst-case $P_\F$ is rank one (i.e., only one state is in-constraint). A larger in-constraint subspace leads to a lower sufficient number of measurements. Moreover, in practice the initial state is unlikely to align perfectly with the worst case presented in Equation~\eqref{eqn:worst_case_state}. 
We also observe in our experiments that the required number of measurements has only a weak dependence on the number of QAOA layers $p$ for the problem instances considered.
Therefore, one could consider a significantly relaxed and simplified version of the rules provided in Corollary~\ref{cor:qaoa_bounds} as follows:
\begin{align}
    \label{eqn:eta_defn}
    N_j = \left\lceil\frac{\beta_j^2}{\eta}\right\rceil,
\end{align}
where $\eta$ is some hyperparameter to be fine tuned. One could always efficiently estimate the in-constraint probability of a QAOA circuit with a fixed $\eta$ by measuring a single auxiliary qubit indicating whether the final state output by the circuit is in-constraint. In the portfolio optimization experiments, we successfully use an $\eta$ for the $B=\sum_j\xgate_j$ mixer that is orders of magnitude larger than predicted by Corollary~\ref{cor:qaoa_bounds}, requiring a correspondingly smaller number of measurements.

\begin{figure}[t]
    \centering
    \vspace{-0.09in}
    \includegraphics{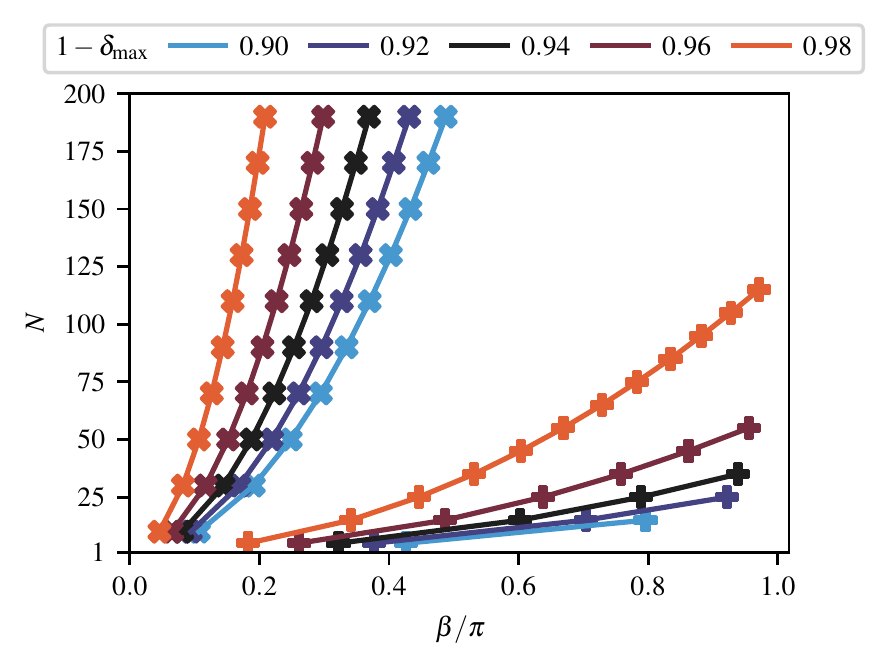}
    \caption{\textbf{Scaling of the number of Zeno measurements.} 
    Number of measurements, obtained from Corollary~\ref{cor:qaoa_bounds}, required in QAOA with Zeno dynamics to maintain a maximum out-of-constraint probability of $\delta_{\max}$ (hence, a minimum  in-constraint probability of  $1 -\delta_{\max}$) for the $B=\sum_j\xgate_j$ (\ding{54} marker) and $B=\ketbra{+}$ (\ding{58} marker) mixers with $3$ qubits. Color denotes the minimum in-constraint probability $1 -\delta_{\max}$, as indicated by the legend. Note that this is the scaling required to ensure the desired minimum in-constraint probability for the worst-case initial state (i.e., Equation~\eqref{eqn:worst_case_state}) and is potentially more pessimistic than what is observed in practice. Many more measurements are required for $B=\sum_j\xgate_j$ as the number of measurements grows quadratically with number of qubits. Note that due to periodicity, the evolution time, $\beta$, can be constrained to $\lvert\beta\rvert \leq \frac{\pi}{2}$ for $B=\sum_j\xgate_j$ and $\lvert \beta \rvert \leq \pi$ for $B=\ketbra{+}$.}
    \label{fig:meas_scale}
\end{figure}

\subsubsection{QAOA with Zeno dynamics in the adiabatic limit}

If the initial state $\ket{s}$ is the ground state of the mixer Hamiltonian $B$, QAOA is known to be able to prepare the ground state of the cost Hamiltonian $C$ and thereby solve the problem exactly in the limit of an infinite number of QAOA layers by approximating adiabatic evolution~\cite{farhi2014quantum}. We now show that this limiting behavior is preserved for constrained QAOA with Zeno dynamics.

Now consider QAOA with constraints enforced by measurement $\calP$ as defined in Equation~\eqref{eq:results_measurement}, in the Zeno limit, when the number of measurements is taken to infinity, the operator describing the asymptotic dynamics is a sum of the original mixer $B$ projected onto the subspaces defined by the projectors constituting $\calP$, i.e., 
\[
    H_{Z} = \calP B = \sum_{j=1}^k P_j B P_j.
\]

Concretely, consider the task of using QAOA to approximate the adiabatic evolution under the following time-dependent Hamiltonian:
\begin{equation}
    H_{s}(t) = (1-s(t))B + s(t)C,
\end{equation}
where $s:[0, T] \xrightarrow[]{}[0, 1]$ is the interpolating schedule function. A common schedule function is the linear schedule defined by
\begin{equation}
    s(t) = \frac{t}{T},
\end{equation}
where $T$ is the evolution time scale. Suppose $T \gg O((\min_{s}\Delta_{n}(s))^{-2})$, where $\Delta_n(s)$ is the instantaneous minimum difference between the $n$-th eigenvalue and any other eigenvalue of $H(s)$. If $\forall s$, it holds that $\Delta_{n}(s) \neq 0$, then the quantum adiabatic theorem \cite{childs2017lecture} implies:
\begin{equation}
    \mathcal{T}\exp(i\int_{0}^{T}H_{s}(t)dt)\ket{\phi_{n}(0)} = \ket{\phi_n(T)}.
\end{equation}

In the Zeno case, we consider 
\begin{equation}
     H_{s}(t) = (1-s(t))H_{Z} + s(t)\mathcal{P}C.
\end{equation}
Consider the QAOA operator with only one measurement per layer, i.e., $\forall j, N_j=1$ in \eqref{eq:qzd_qaoa_circuit}:
\begin{equation}
    \U(p) = \prod_{j=1}^{p}\mathcal{P}\mathcal{U}_{B}\left(\beta_j\right)\mathcal{U}_{C}\left(\gamma_j\right).
\end{equation}

Now it is easy to recover the parameters $\beta_j$, $\gamma_j$ giving the limit. From the definition of the product integral \cite{dollard_friedman_1984} it follows that
\begin{align}
    &\mathcal{T}\exp(i\int_{0}^{T}H_{s}(t)dt) \nonumber \\ 
    =& \lim_{p \to \infty} \prod_{j=1}^{p}\exp(i\frac{T}{p}H_{s}\left(\frac{jT}{p}\right)) \nonumber \\
    =& \lim_{p \to \infty} \prod_{j=1}^{p}\exp(i\frac{T}{p}\left[\left(1 - \frac{j}{p}\right)\mathcal{P}B + \left(\frac{j}{p}\right)\mathcal{P}C\right]) \nonumber \\
    =& \lim_{p \to \infty} \prod_{j=1}^{p}\mathcal{P}\exp(i\frac{T}{p}\left(1-\frac{j}{p}\right)B)\exp(i\frac{jT}{p^2}C),
\end{align}
where the third equality follows from expanding to the first order in $\frac{T}{p}$ and that $\frac{j}{p}$ and $1 - \frac{j}{p}$ are bounded by $1$. Also, since the evolution is in a finite-dimensional space, $B$ and $C$ have bounded operator norms.

Thus if $\rho_n(0) = \ketbra{\psi_n(0)}$ is an $n$-th eigenstate of $H_{Z}$ then
\begin{equation}
    \rho_n(T) = \lim_{p \to \infty}\U(p)\rho_n(0),
\end{equation}
where $\rho_n(T)$ is pure and is an $n$-th eigenstate of $\mathcal{P}C$. Thus with $\beta_j = -\frac{T}{p}\left(1-\frac{j}{p}\right)$ and $\gamma_j = -\frac{jT}{p^2}$ as $p \xrightarrow[]{} \infty$, QAOA with Zeno dynamics approaches the adiabatic limit and recovers the optimal solution.

\subsubsection{Mitigating mixer limitations in the Zeno limit} \label{sec:mixer_limitations}

While the evolution under $P_{\mathcal{F}}BP_{\mathcal{F}}$ is guaranteed to preserve the in-constraint subspace, it may inhibit transitions between states in $\mathcal{F}$ that were allowed with $B$. This is because states in $\mathcal{F}$ may be connected by $B$ through a path that passes through states not in $\mathcal{F}$. To see this, consider a simple example of the two-qubit mixer $B_{2} = \xgate_{1} + \xgate_{2}$ and the in-constraint space $\mathcal{F} = \{\ket{01}, \ket{10}\}$.
In the Zeno limit, the mixing operator evolution in the in-constraint subspace is generated by $P_{\mathcal{F}}B_{2}P_{\mathcal{F}}$, which equals the zero matrix. Thus, the propagator corresponding to the projected mixer becomes the identity operator and the dynamics become trivial. In general, if there is no path between two computational basis states $\ket{j}, \ket{k} \in \mathcal{F}$ in the graph defined by $B$, the continuous-time quantum walk defined by the mixing operator cannot move probability amplitude from $\ket{k}$ to $\ket{j}$. Whether the transitions between in-constraint states are suppressed in the Zeno limit is in general dependent on the in-constraint space $\F$.

One way to avoid the issue of suppressed transitions is by choosing a mixer $B$ with a complete connectivity graph among computational basis states, i.e., $B = \ketbra{+}$. This mixer is also known as the complete-graph mixer~\cite{McClean_2021, bartschi2020grover}. It has been conjectured~\cite{McClean_2021} that mixers with high connectivity, such as the $B = \ketbra{+}$, can at best produce a Grover-like speedup since they do not make use of the structure of the cost operator. While it is unclear if this conjecture is true,
we emphasize that our approach can utilize any mixer and can efficiently enforce constraints as long as the difference between the maximum and minimum eigenvalues of the mixer is polynomial in the number of qubits.

\subsection{Numerical Experiments}\label{sec:numerical_experiments}

We now present the numerical experiments showing the power of the proposed method. The technique we propose is general, though in this section we consider only the problem of portfolio optimization (with both equality and inequality constraints) and only the QAOA and L-VQE algorithms. {The parameters in QAOA and VQE were optimized using COBYLA \cite{Powell1994} initialized with a large number of random initial points.} We compare the results to the state-of-the-art method of encoding constraints by introducing a penalty into the objective, and observe significant improvements in both approximation ratio and in-constraint probability. In addition to better performance, the proposed method does not require complicated tuning of the penalty factor.

\subsubsection{Benchmark: portfolio optimization}\label{sec:portfolio_opt}
The daily operation of a large financial institution requires solving many classically-hard optimization problems~\cite{herman2023,yalovetzky2021, he2023alignment}. Among such problems, one of the most important is portfolio optimization. Modern portfolio theory~\cite{markowitz1952harry} considers the task of finding a portfolio with a desired trade-off between risk and expected return. This task is typically formulated as an optimization problem, which is hard to solve classically in many settings, such as when the variables are required to only take on a discrete set of values. When designing an algorithm for portfolio optimization, a central consideration is the ability to incorporate a general class of constraints. Such constraints can come from regulatory or business considerations, with examples ranging from portfolio-level constraints (including budget and total number of assets) to asset-level constraints (such as minimum holding size).

The particular constrained portfolio optimization problems we study numerically arise from the discrete mean-variance Markowitz model~\cite{markowitz1952harry} and have the following objective function
\begin{equation}
    \min_{\bm{x}\in \F} q\bm{x}^{\mathsf{T}}\Sigma\bm{x}-\bm{\mu}^{\mathsf{T}}\bm{x},
\end{equation}
where $\F$ is defined by some set of constraints on the portfolio. We consider two sets of problems. In the first set, we impose an inequality constraint on the total size of the portfolio ($\sum_j x_j \leq C$). In the second set of problems, in addition to the inequality constraint on portfolio size, we include a constraint on the total expected return ($\sum_j \mu_jx_j \ge R$). For each of the two sets of constraints, we consider seven instances with between four and ten assets, for a total of fourteen instances. In all problem instances $\mathcal{F} \subset \mathbb{B}^n$, where $n$ is the number of assets.

\subsubsection{Zeno dynamics improves quantum optimization performance}
Figure~\ref{fig:zeno_vs_penalty} presents the comparison between QAOA with Zeno dynamics and QAOA with constraints enforced using a penalty factor on the fourteen problem instances described in the previous subsection. The penalty method is described in  \alt{the \emph{Methods} Section}{Section~\ref{sec:preliminaries}}. The solution quality is measured in terms of the approximation ratio $r$, a value between $0$ and $1$, with larger $r$ being better. The approximation ratio is  formally defined in \alt{the \emph{Methods} Section}{Section~\ref{sec:preliminaries_qopt}}. We consider QAOA with mixers $B=\sum_j\xgate_j$ (\ding{54} marker) and $B=\ketbra{+}$ (\ding{58} marker), and optimize the QAOA parameters exhaustively. To improve the performance of parameter optimization, we follow Ref.~\cite{montanaro2022peptide} and rescale the cost function so that the gradients with respect to $\boldsymbol{\beta}$ and $\boldsymbol{\gamma}$ are roughly of the same magnitude.

\begin{figure}[ht]
    \centerfloat  %
    \includegraphics{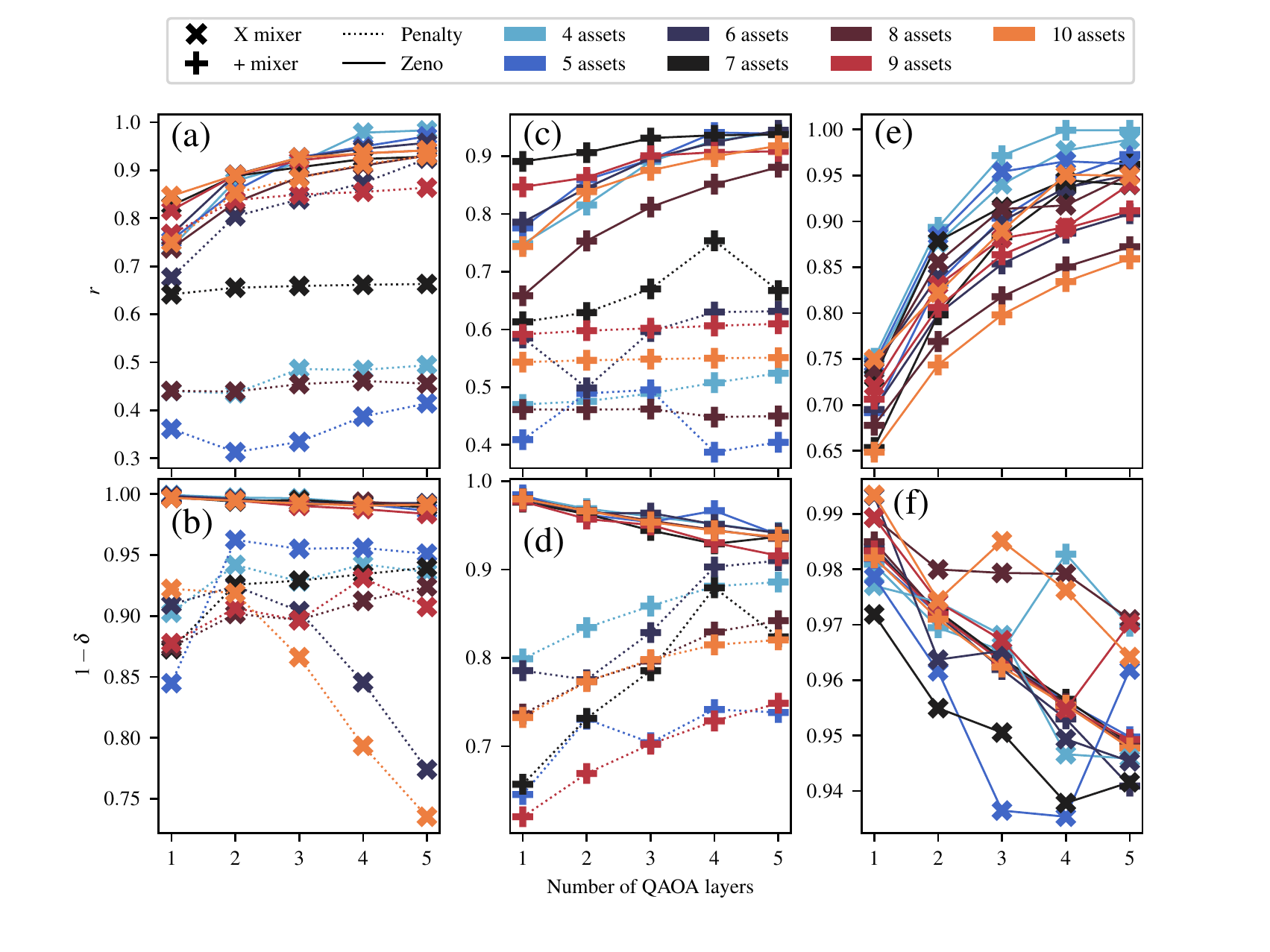}
    \caption{\textbf{Performance of QAOA with Zeno dynamics and QAOA with constraints enforced using penalty terms.} 
    Approximation ratio $r$ and out-of-constraint probability $\delta$ (correspondingly $1-\delta$ in-constraint probability) achieved by QAOA with constraints enforced using penalty terms (dotted lines) on problems (a,b,c,d) with a single constraint, and by QAOA with Zeno dynamics (solid lines)
    on problems with a single (a,b,c,d) and  multiple (e,f) constraint(s). The markers \ding{54} and \ding{58} indicate whether QAOA used the    $B=\sum_j\xgate_j$ mixer or $B=\ketbra{+}$  mixer, respectively. For all single constraint problems, QAOA with Zeno dynamics produces a superior approximation ratio and in-constraint probability (solid line is above dotted line with the same color). As penalty factor tuning is prohibitively difficult for problems with multiple constraints (see \alt{the \emph{Results} Section}{Section~\ref{sec:penalty_factor_tune}}), for these problems only Zeno dynamics results are presented.}
    \label{fig:zeno_vs_penalty}
\end{figure}

For instances with a single constraint (see dotted lines in Figure~\ref{fig:zeno_vs_penalty}(a,b,c,d)) we perform  extensive tuning of the penalty factor $\lambda$. For multi-constraint problems, the tuning becomes prohibitively expensive. Therefore, we exclude QAOA with constraints enforced through penalties from the comparison for problems with multiple constraints. The choice of the penalty factor and the difficulty of its optimization are discussed in detail in the next subsection.

We observe that Zeno dynamics (see solid lines in Figure~\ref{fig:zeno_vs_penalty}(a,b,c,d)) enables consistently better solution quality and in-constraint probability as compared to QAOA with constraints enforced using a penalty (dotted lines) for all problems considered. Furthermore, Figure~\ref{fig:zeno_vs_penalty}b shows that for 6 and 10 assets the in-constraint probability drops off rapidly with the number of QAOA layers if the penalty factor is kept constant.  This highlights an important limitation of enforcing the constraints via penalties, namely that the penalty factor must be tuned independently for each QAOA depth. In contrast, for QAOA with Zeno dynamics we obtain an explicit rule for how $\eta$, from \eqref{eqn:eta_defn}, should change with the QAOA depth (see Corollary~\ref{cor:qaoa_bounds}). However, for the numerics shown in Figure~\ref{fig:zeno_vs_penalty}, we fix $\eta$ to ensure a constant minimum in-constraint probability per layer. We observe good performance despite $\eta$ being a depth-independent constant in this case. We note that since $\eta$ was held constant
while $p$ varied, the in-constraint probability slowly decreases with the number of layers as predicted by Corollary \ref{cor:qaoa_bounds}. For $B=\ketbra{+}$ mixer, this results in an average number of measurements of $\approx 77$ for 6 assets and $\approx 35$ for 7 assets.

Since multiple constraints can be efficiently handled in the Zeno framework, in Figure~\ref{fig:zeno_vs_penalty}(e,f), we include the performance of QAOA with Zeno dynamics on problems with multiple constraints (one on the budget and one on the total expected return). The results show that the Zeno-enhanced QAOA is able to achieve a similar performance as it did for the single-constraint problems, with sufficiently high $p$.

We note that the in-constraint probability can be improved arbitrarily for the Zeno dynamics approach by decreasing $\eta$, without the need to re-optimize the QAOA parameters. This is due to the objective function landscape becoming independent of $\eta$ as the Zeno limit is approached. In fact, we observe that transferring parameters from a smaller to a larger number of measurements (larger to smaller $\eta$) works well even for practically relevant values of $\eta$. Figure~\ref{fig:eta_sweep_transf} shows the approximation ratio $r$ and in-constraint probability with directly optimized QAOA parameters and with pre-optimized parameters transferred from a fixed value of $\eta=1.6$ (marked with a star in the plot). We observe that for sufficiently small $\eta$, transfer works well and the difference in approximation ratio is negligible. Specifically,  parameter transfer using the  $B=\sum_j\xgate_j$ mixer and a \emph{total} of $33$, $75$, and $200$ measurements results in in-constraint probabilities of at least $85\%$, $89\%$, and $96\%$, respectively for the nine-assets, single-constraint problem at $p=5$. At the same time, if the number of measurements is very small ($\eta$ large), the objective function landscape is very different from the landscape in the Zeno limit, and the parameter transfer does not work well. We remark that while the in-constraint probability increases monotonically as $\eta$ decreases, no such guarantee is given for approximation ratio $r$. In fact, in Figure~\ref{fig:eta_sweep_transf} we observe that depending on the problem and the circuit depth, $r$ can either increase or decrease with $\eta$.

\begin{figure}
    \centering
    \includegraphics{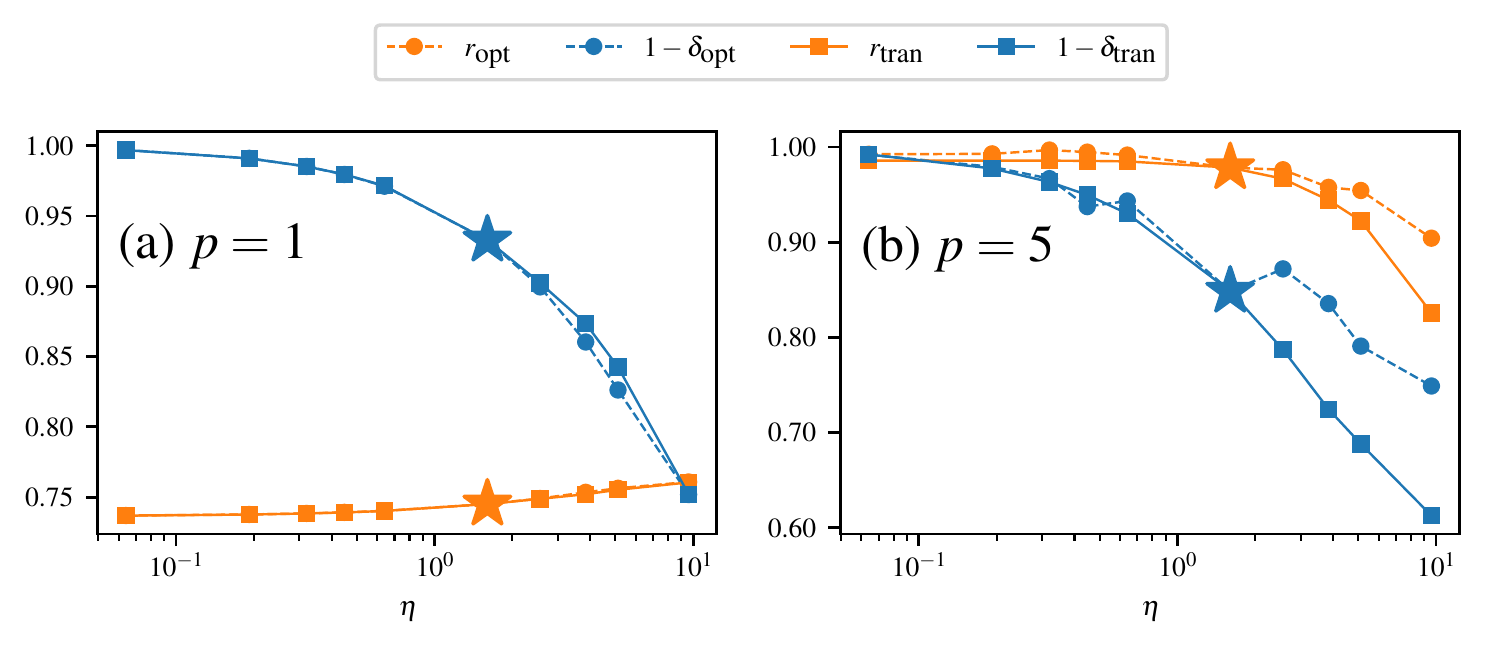}
    \caption{\textbf{Transferability of parameters in QAOA with Zeno dynamics.} 
    Performance of a $1$-layer (a) and $5$-layer (b) QAOA with Zeno dynamics and mixer $B = \sum_j\xgate_j$ with directly optimized parameters ($r_{\text{opt}}$, {$1 - \delta_{\text{opt}}$}) and with parameters transferred from a fixed value of $\eta=1.6$ ( $r_{\text{tran}}$, {$1-\delta_{\text{tran}}$}). The source is marked with a star. Corresponding to each case, $r$ signifies the approximation ratio and $\delta$ the out-of-constraint probability. For values of the hyperparameter $\eta$, which controls the number of measurements and is defined in Equation~\eqref{eqn:eta_defn}, smaller than $1.6$, the difference between performance with optimized and transferred parameters is negligible (dashed line very close to the solid line).}
    \label{fig:eta_sweep_transf}
\end{figure}

Note that the same approach of boosting the in-constraint probability without re-optimizing the QAOA parameters does not work if the constraints are enforced using penalties. Figure~\ref{fig:penalty_param_transfer} shows that transferring parameters from a fixed value of penalty factor (marked with a star) leads to the approximation ratio rapidly dropping off to random guess. It is however possible that better performance may be achieved by leveraging more sophisticated parameter transfer strategies, such as the rescaling rule proposed for the weighted MaxCut problem~\cite{2201.11785, sureshbabu2023parameter} or machine learning methods~\cite{khairy2019learning}.

\begin{figure}
    \centering
    \includegraphics{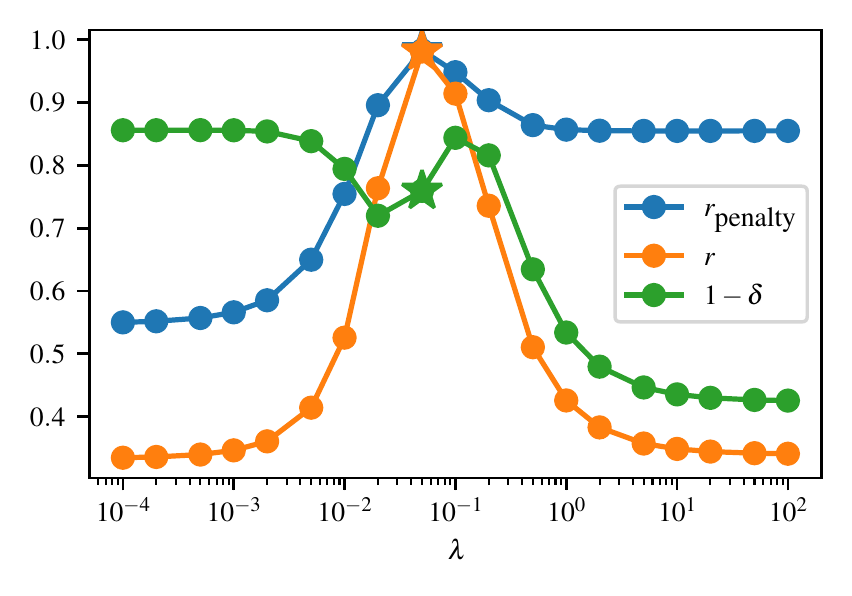}
    \caption{\textbf{Transferability of parameters in QAOA with penalty terms.}
    Performance of QAOA with $B = \sum_j\xgate_j$ mixer and constraints enforced through penalties with parameters transferred from a fixed value of penalty factor $\lambda=0.1$ (source marked with a star). The out-of-constraint probability is $\delta$. The approximation ratio $r$ (Equation \eqref{eq:approx_ratio_vqe}), unlike $r_{\text{penalty}}$ (Equation \eqref{eq:approx_ratio_penalty_vqe}), excludes the penalty objective and drops off to random guess if transferring parameters to values of $\lambda$ sufficiently different from source.}
    \label{fig:penalty_param_transfer}
\end{figure}

\begin{figure}[ht]
    \centering
    \includegraphics{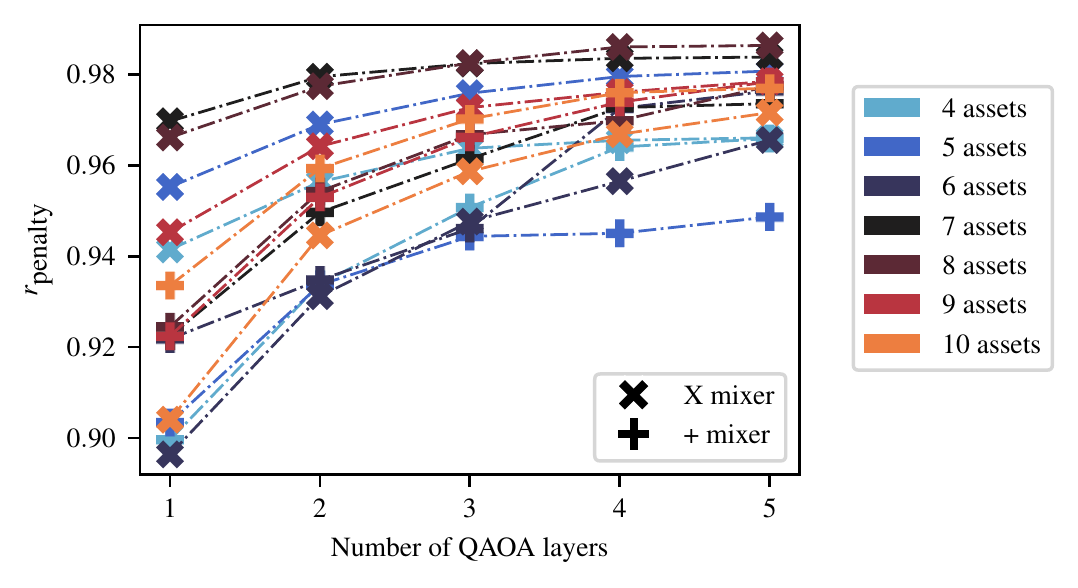}
    \caption{\textbf{Approximation ratio of QAOA with penalty terms using different numbers of QAOA layers.} 
    The approximation ratio (as defined in Equation~\eqref{eq:approx_ratio_penalty_vqe}) for the full objective with penalty terms increases monotonically with the number of QAOA layers, as expected.  However, the in-constraint approximation ratio (as defined in Equation~\eqref{eq:approx_ratio_vqe}) is not guaranteed to change monotonically, as seen in Figure~\ref{fig:zeno_vs_penalty}a,c. Color denotes the number of assets in the optimization problem, as shown in the legend. The markers \ding{54} and \ding{58} indicate whether QAOA used the    $B=\sum_j\xgate_j$ mixer or $B=\ketbra{+}$  mixer, respectively.}
    \label{fig:penalty_single_constraint_ar}
\end{figure}

While for QAOA with Zeno dynamics the approximation ratio $r$ given in Equation~\eqref{eq:approx_ratio_vqe} increases monotonically with the number of QAOA layers, this is not guaranteed for QAOA with constraints enforced through penalties. This is because the QAOA parameters are chosen with respect to the objective with penalties and the increased expressivity of the higher-depth circuit is only guaranteed to improve the performance with respect to that objective. Figure~\ref{fig:penalty_single_constraint_ar} shows that this is indeed the case and the approximation ratio $r_{\text{penalty}}$ given in Equation~\eqref{eq:approx_ratio_penalty_vqe} increases with the number of QAOA layers as expected.

Finally, we include the results for Zeno-enhanced L-VQE with $L=1$ in Equation~\eqref{eqn:ansatz_vqe_with_meas}. The structure of L-VQE is presented in Equation \eqref{eq:ansatz_lvqe} and further described in \alt{the \emph{Methods} Section}{Section~\ref{sec:preliminaries_qopt}}.  However, instead of using Corollary~\ref{cor:scheme_2_cor} to determine a sufficient value for the number of measurements $N$, we heuristically set  $N= 100$. Table~\ref{tab:lvqe} presents the results. As expected, L-VQE achieves high approximation ratio, while Zeno dynamics enables high in-constraint probability. As the total number of measurements is kept fixed for all problems and parameter values, slightly lower in-constraint probability is observed for higher qubit counts. As is the case for QAOA, the in-constraint probability can be increased by increasing the number of measurements.

\begin{table}[h]
    
    \centering
    \resizebox{0.5\textwidth}{!}{
    \begin{tabular}{c|c|c|c|c}
          \multirow{2}{*}{\# assets} & \multicolumn{2}{c|}{Single} & \multicolumn{2}{c}{Multiple}  \\
         & $r$ & ${1-\delta}$ &  $r$ & ${1 - \delta}$ \\
         \hline
4 & 0.995 & 0.964 & 0.9996 & 0.980 \\
5 & 0.995 & 0.913 & 0.977 & 0.909 \\
6 & 0.972 & 0.895 & 0.964 & 0.963 \\
7 & 0.979 & 0.870 & 0.917 & 0.936 \\
8 & 0.956 & 0.887 & 0.948 & 0.944 \\
9 & 0.967 & 0.844 & 0.961 & 0.974 \\
10 & 0.914 & 0.811 & 0.910 & 0.960        
    \end{tabular}
}
    \caption{\textbf{Performance of L-VQE with Zeno dynamics on the benchmark problems.} Layer variational quantum eigensolver (L-VQE) enhanced with Zeno dynamics obtains high approximation ratio $r$ (Equation \eqref{eq:approx_ratio_vqe}) and high in-constraint probability $1-\delta$. The algorithm was applied to both problems with a single constraint and multiple constraints. The Zeno-enhanced L-VQE circuit was constructed by inserting measurements after all parameterized gates have been applied. This corresponds to $L=1$ in Equation~\eqref{eqn:ansatz_vqe_with_meas}. The number of measurements was heuristically set to $100$.}
    \label{tab:lvqe}
\end{table}

\subsubsection{Penalty factor tuning is difficult} \label{sec:penalty_factor_tune}

An important advantage of our method is the simplicity of hyperparameter tuning, as only $\eta$ in Equation~\eqref{eqn:eta_defn} needs to be chosen. This choice is made easy by Theorem~\ref{thm:main_theorem} and its corollaries, which 
imply the monotonic increase of in-constraint probability with decrease in $\eta$. This is in sharp contrast with the penalty approach, where the performance crucially depends on the penalty strength, which is hard to tune in general. We now present how the penalty strength was chosen for the experiments above, and discuss the challenges that arose in doing so.

\begin{figure}
    \centering

    \includegraphics{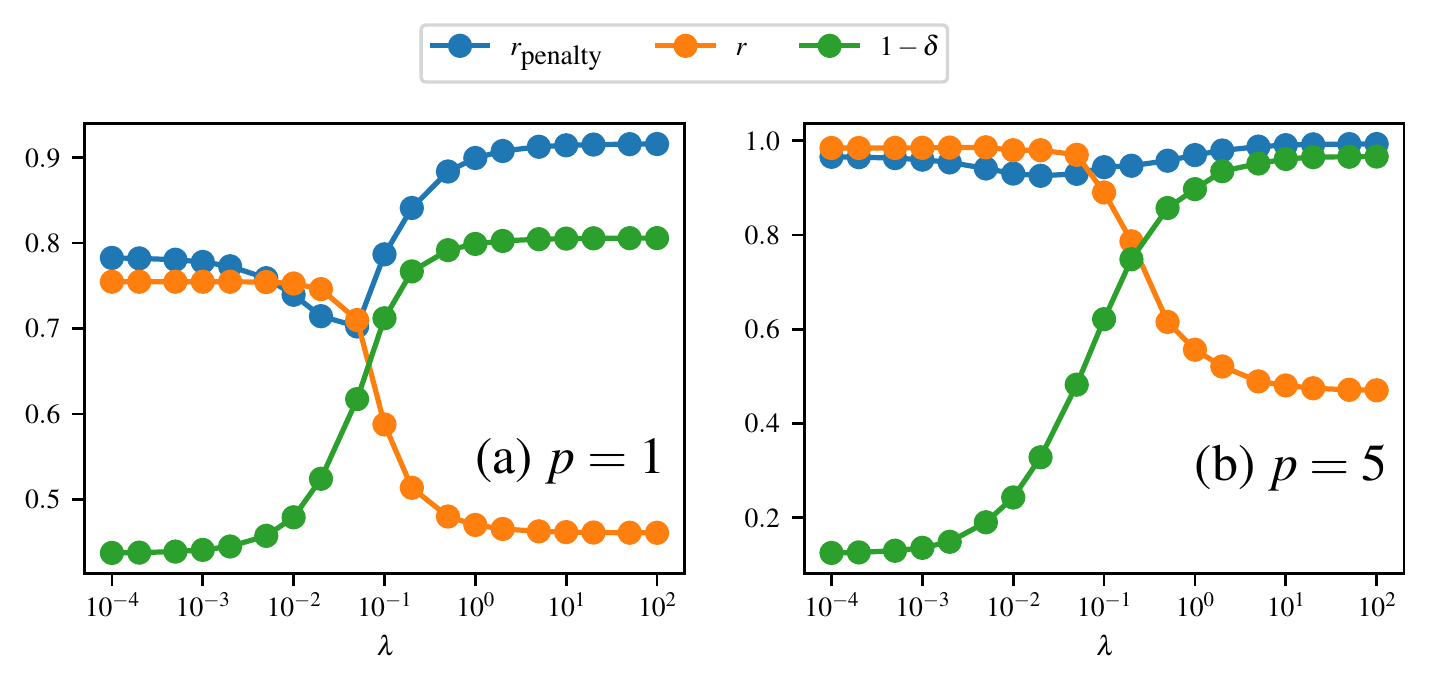}
    \caption{\textbf{Difficulty of penalty factor tuning for QAOA with a single penalty term.} 
    Performance of QAOA with a single constraint enforced through a penalty term with varying penalty factors $\lambda$. A trade-off occurs between the approximation ratio $r$ (Equation \eqref{eq:approx_ratio_vqe}) and the in-constraint probability {$1 - \delta$}. As shown in (a), the maximum in-constraint probability is limited by the expressivity of the QAOA circuit at low depth ($1$ QAOA layer, or $p=1$). With $5$ layers (b),  QAOA is able to achieve better performance in terms of the penalized objective, as indicated by the approximation ratio $r_\text{penalty}$ (Equation \eqref{eq:approx_ratio_penalty_vqe}). However, there is still a significant trade-off between the true objective $r$ and in-constraint probability.}\label{fig:penalty_factor_sweep_1d}
\end{figure}

\begin{figure}
    \centering
    \includegraphics{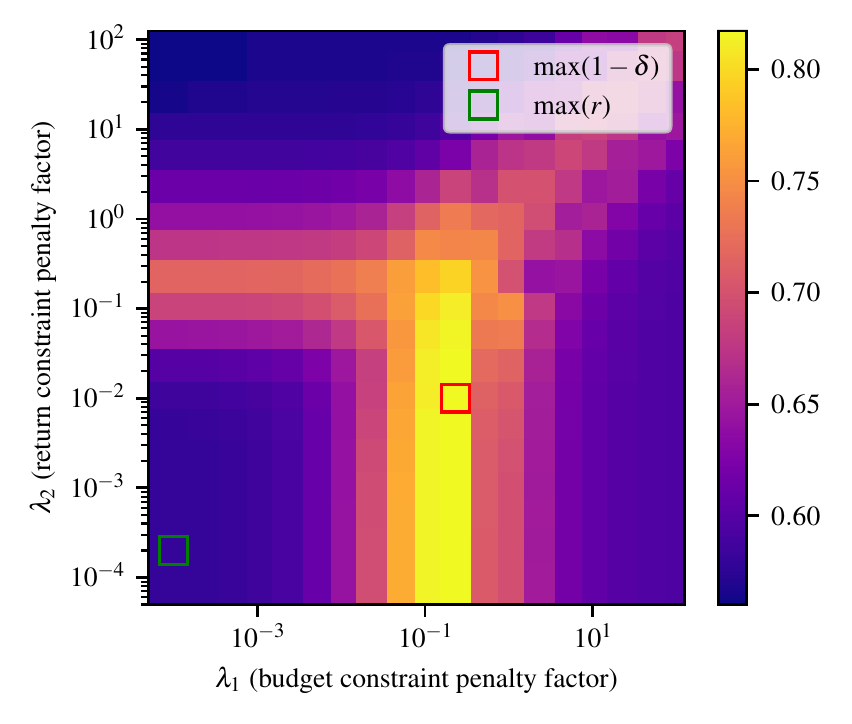}
    \caption{\textbf{Difficulty of penalty factor tuning for QAOA with two penalty terms.} 
    In-constraint probability of optimized solution using QAOA applied to an objective with two penalty functions,  associated with separate constraints. The corresponding penalty factors are indicated by $\lambda_1$ and $\lambda_2$, respectively. One is a maximum budget constraint and the other is  minimum return constraint. The value $\delta$ is the out-of-constraint probability, and $r$ is the approximation ratio (as defined in Equation~\eqref{eq:approx_ratio_vqe}). The square highlighted in red corresponds to the maximum in-constraint probability ($1-\delta$) over all combinations of the two penalty factors, and the square highlighted in green corresponds to the maximum $r$. This highlights that both large in-constraint probability and large approximation ratio cannot be obtained. 
 The figure shows results for the $B=\ketbra{+}$ mixer and $3$-layer QAOA, though we observe similar behavior for all mixers and QAOA depths considered. }
    \label{fig:penalty_factor_sweep_2d}
\end{figure}

Figure~\ref{fig:penalty_factor_sweep_1d} presents the performance of QAOA on a single-constraint problem enforced using a penalty term with varying penalty factors $\lambda$. 
In the plot, the in-constraint probability {$1 - \delta$} monotonically increases with $\lambda$, while the approximation ratio $r$ decreases.
This indicates a trade-off between $r$ and the out-of-constraint probability $\delta$, and hence hyperparameter tuning on $\lambda$ must be performed in order to obtain a good approximation ratio while meeting requirements on the minimum in-constraint probability.
We also observe that for QAOA with small $p$, $1-\delta$ tends to levels off at a value far below what is achievable by using Zeno dynamics. 
For example, the top figure in Figure~\ref{fig:penalty_factor_sweep_1d} shows that the highest in-constraint probability achievable with $p=1$ is around $80\%$ for the problem tested.
Given that the approximation ratio with the penalty term $r_\text{penalty}$ is above $0.9$ for the high $\lambda$ regime, it indicates that the maximum achievable in-constraint probability may be limited by the expressivity of the variational circuit. 
On the other hand, constraints enforced by Zeno dynamics do not suffer from such problems, as the in-constraint probability can be arbitrarily boosted regardless of the expressivity of the varational circuit (see Figure~\ref{fig:eta_sweep_transf}). In the numerical experiments, we choose the value of $\lambda$ independently for each problem instance with the goal of obtaining a high in-constraint probability $1-\delta$. Since we show that the factor $\lambda$ trades off $r$ and $\delta$, both cannot be improved at the same time. This suggests that there does not exist a choice of $\lambda$ such that QAOA with the penalty method outperforms QAOA with Zeno dynamics.

\begin{figure*}[!tb]
    \centering
    \includegraphics[width=1.0\textwidth]{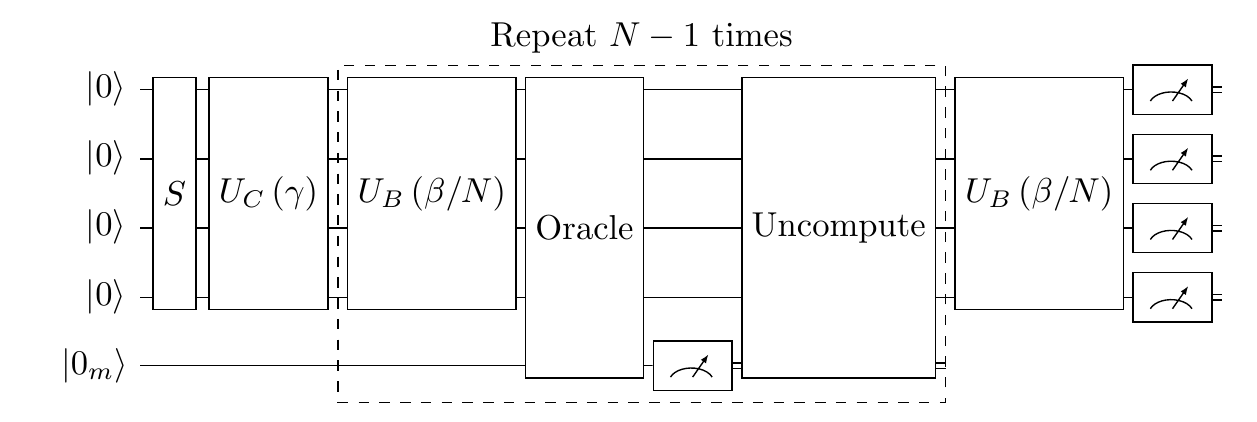}
    \caption{\textbf{Quantum circuit for QAOA with Zeno dynamics.}
    QAOA circuit with Zeno dynamics used in hardware runs for one-layer QAOA ($p=1$) on four-asset problems. The operator $S$ prepares a uniform superposition over feasible states.}
    \label{fig:zeno_hardware_qaoa_circuit}
\end{figure*}

\begin{figure*}[!tb]
    \centerfloat  %
    \subfloat[]{\includegraphics[width=1.0\textwidth]{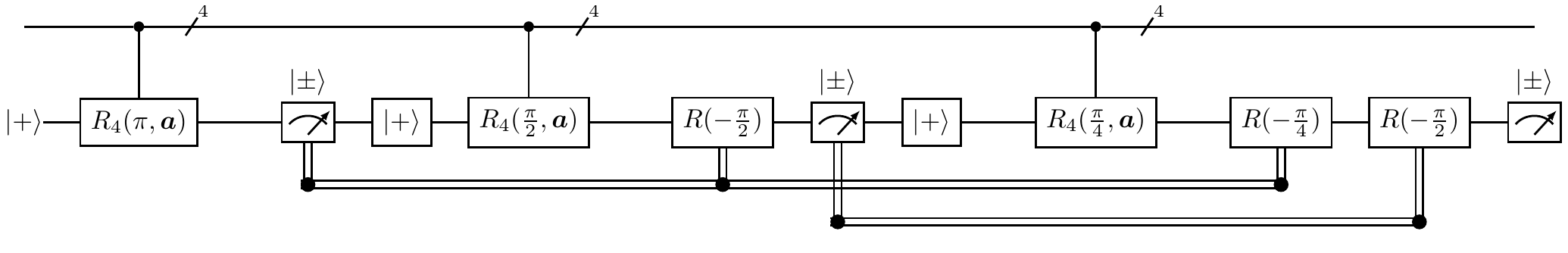}}\\
    \subfloat[]{\includegraphics[width=0.6\textwidth]
    {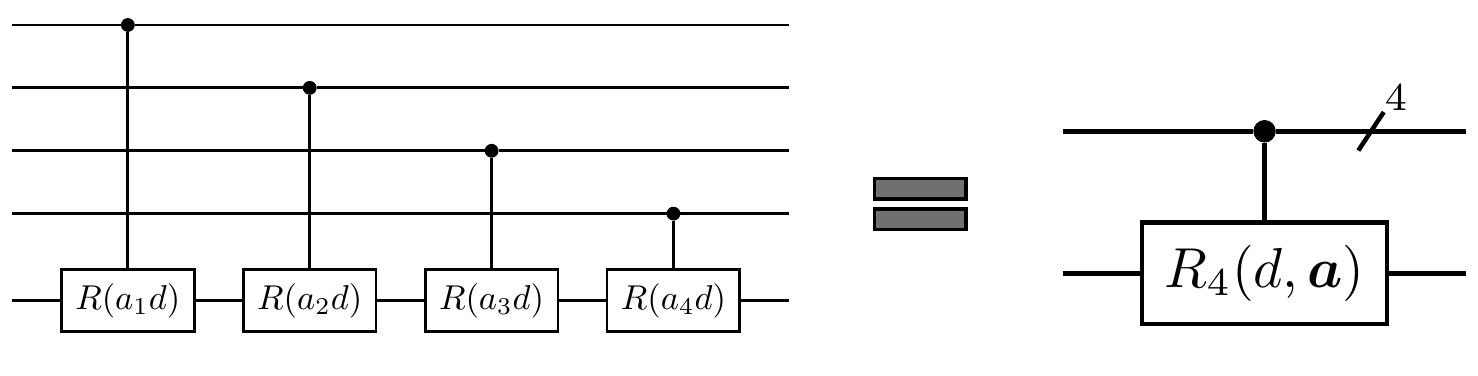}}
    \caption{\textbf{QFT adder with QCL.} 
    Quantum circuits for (a) semiclassical quantum Fourier transform adder with quantum conditional logic (QCL) used in the hardware experiments involving equality constraints and (b) the four-qubit rotation gate used in (a). Note that $R(\alpha)$ denotes a phase gate. For the  equality-constraint experiment executed on the  H1-2 quantum device, we set  $\bm{a} := (a_1, a_2, a_3, a_4) = (2, -1, -1, 0)$. The uncomputation step consists of resetting the auxiliary qubit to the $\ket{+}$ state.}
    \label{fig:qcl_fourier_oracle}
\end{figure*}

For problems with multiple constraints, hyperparameter tuning should generally be performed on each penalty factor $\lambda_j$ included in the relaxed objective (Equation~\eqref{eq:obj_relaxed}).
This means that hyperparameter tuning can quickly become infeasible, as the search space for all $\lambda_j$'s grows exponentially with the number of penalty terms.
We show in Figure~\ref{fig:penalty_factor_sweep_2d} how hyperparameter tuning works with two penalty factors: $\lambda_1$ and $\lambda_2$, which correspond to penalty terms enforcing the budget constraint and the return constraint respectively.
The figure shows the in-constraint probability of the optimal solution obtained with varying $\lambda_1$ and $\lambda_2$.
Similar to the single-constraint case, maximal approximation ratio $r$ and maximal in-constraint probability $1-\delta$ cannot be simultaneously achieved.
Specifically, the solutions with the maximal $r$ and maximal {$1 - \delta$} have very different values in $\lambda_1$ and $\lambda_2$.
Moreover, unlike Figure~\ref{fig:penalty_factor_sweep_1d}, Figure~\ref{fig:penalty_factor_sweep_2d} clearly shows the non-monotonic behavior of $1-\delta$ in both $\lambda_1$ and $\lambda_2$.
In fact, we observe a similar behavior across many of the single- and multi-constraint problems that we have tested, and for both the $B=\sum_j\xgate_j$ and $B=\ketbra{+}$ mixers.
This indicates that tuning the penalty factors is indeed difficult in the general case.

\subsection{Hardware Experiments}\label{sec:hardware_experiments}
{While the numerical experiments presented earlier show evidence of the performance of our technique, they do not make use of any concrete circuit implementations of the constraint-checking oracles. In this section, we consider optimized circuit implementations of constraint-checking oracles for two proof-of-concept portfolio optimization problems on noisy quantum hardware. This enables us to validate all of the hardware features, such as mid-circuit measurements and quantum conditional logic (QCL), that are required to implement the efficient oracle construction presented in \alt{the \emph{Methods} Section}{Section~\ref{sec:implement_oracle}}.} 

We execute QAOA with Zeno dynamics on the Quantinuum H1-2 trapped-ion quantum processor. Our implementation uses constraint-checking oracles that perform quantum arithmetic in the Fourier domain, following directly the construction in \alt{the \emph{Methods} Section}{Section~\ref{sec:implement_oracle}}. We observe that increasing the number of measurements improves the in-constraint probability {$1 - \delta$}, as expected. The improvement from additional measurements continues up to a two-qubit gate depth of 148, at which point the hardware noise prevents further improvements. 

The experiments presented in this Section utilize $p=1$ QAOA and the $B=\sum_j\xgate_j$ mixer. We use the cost function of the four-assets portfolio optimization problem used in the numerics described in \alt{the \emph{Results} Section}{Section~\ref{sec:numerical_experiments}}, but apply different constraints.  We consider two instances with linear constraints, one with an equality constraint and one with an inequality constraint. Figure~\ref{fig:zeno_hardware_qaoa_circuit} shows a high-level circuit diagram. For each problem, the QAOA parameters are first optimized using a noiseless simulator. All circuit executions use $2000$ shots and no error mitigation.

The first portfolio optimization instance we consider has an equality constraint on the four binary variables $x_1, x_2, x_3, x_4$: $2x_1 - x_2 - x_3 = 0$. As discussed in \alt{the \emph{Methods} Section}{Section~\ref{sec:quantum_fourier_arithmetic}}, the semiclassical quantum Fourier transform (QFT)  can be utilized for equality constraints. The semiclassical QFT makes use of QCL and midcircuit measurements, which are features supported by the H1-2 device. This results in an oracle that uses only one auxiliary qubit, and thus the circuit uses five qubits in total. The circuit for the oracle is shown in Figure~\ref{fig:qcl_fourier_oracle}. We note that the uncomputation step consists of resetting the one auxiliary qubit to the $\ket{+}$ state.

\begin{figure*}[!tb]
    \centering
    \subfloat[]{\includegraphics[width=0.6\textwidth]{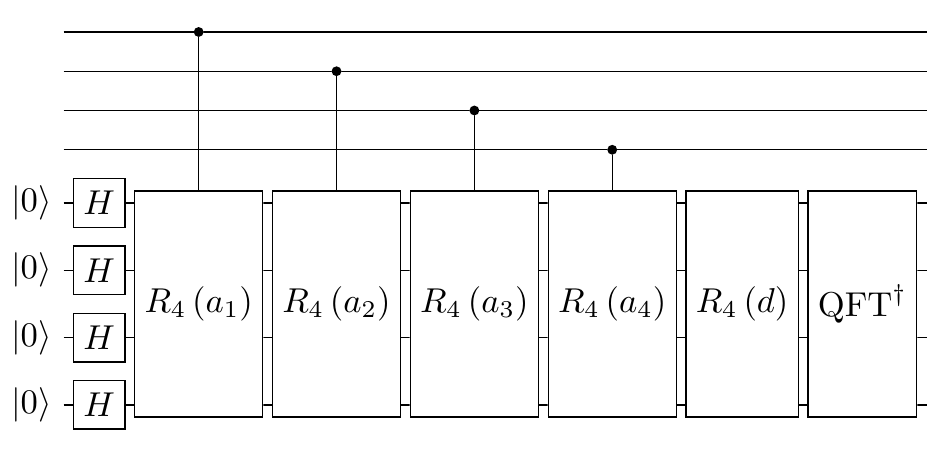}}
    \hspace{0.3in}
    \subfloat[]{\includegraphics[width=0.3\textwidth]{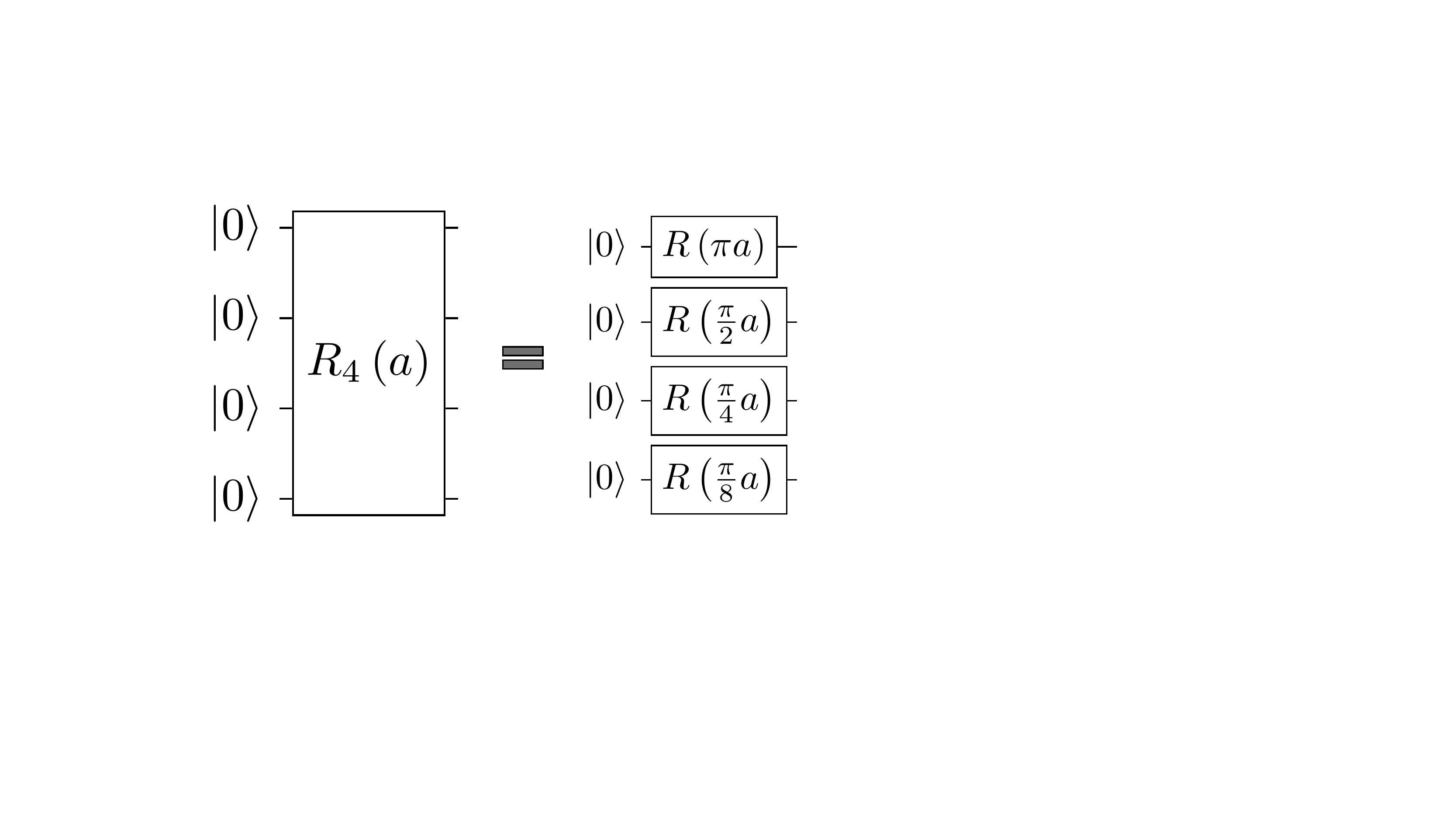}}
    \caption{
    \textbf{QFT adder without QCL.}
    Quantum circuits for (a) the quantum Fourier transform (QFT) adder used in the hardware experiments and (b) the four-qubit rotation gate used in (a). Note that $R(\alpha)$ denotes a phase gate. For the inequality-constraint  experiment, we set $a_1 = a_2 = a_3 = a_4 = 1, d = -3$ and used four qubits for precision. For the equality-constraint experiment, without quantum conditional logic (QCL), we set  $a_1 = 2, a_2 = a_3 = -1, a_4 = d = 0$ and used only three qubits for precision. For the inequality constraint, the inverse of the oracle is applied after measuring the qubit encoding the sign. However, for the equality constraint, since all auxiliary qubits are measured, we do not need to apply the inverse QFT operator and can simply reset all auxiliary qubits to the ground state.
    Note that here the inverse QFT operator ($\text{QFT}^{\dagger}$) does not include swaps as the reordering has been done by rearranging the banks of controlled rotations.}
    \label{fig:qft_fourier_oracle}
\end{figure*}

\begin{figure*}
    \centerfloat  %
    \includegraphics[width=\textwidth]{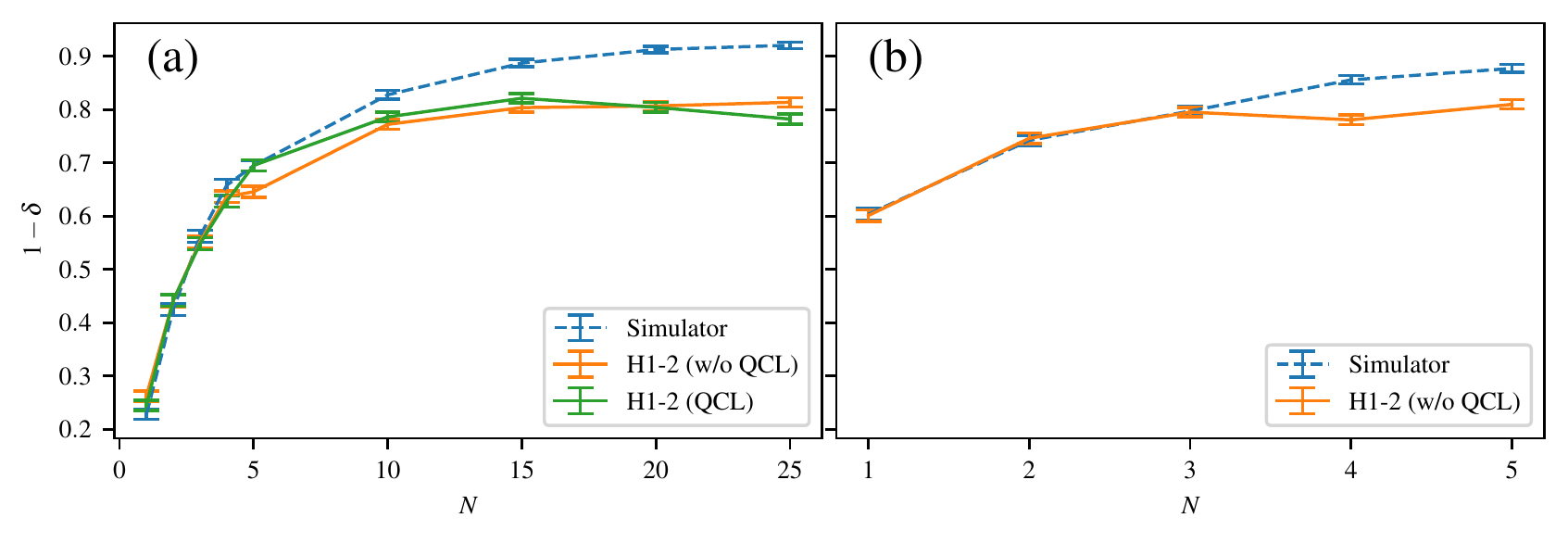}
    \caption{\textbf{Simulation and hardware experiment results using QAOA with Zeno dynamics.}
    QAOA with $p=1$ and Zeno dynamics was applied to solve a four-asset problem with an equality constraint $2x_1 - x_2 - x_3 = 0$ (a) and inequality constraint $\sum_{j=1}^4x_j \leq 2$ (b). The circuits were executed on a classical simulator and on the H1-2 quantum device. The oracles are implemented using arithmetic in the Fourier domain.  For the equality constraint (a), the quantum conditional logic (QCL) implementation of the Fourier adder used one auxiliary qubit, and the version without QCL used three auxiliary qubits. The Fourier adder used for the inequality constraint (b) used four auxiliary qubits. Error bars indicate the standard error of the mean arising from finite sampling ($2000$ shots). The in-constraint probability {$1-\delta$} grows with the number of measurements ($N$).}
    \label{fig:hardware_delta_vs_N}
\end{figure*}

\begin{figure*}
    \centerfloat  %
    \subfloat[]{\includegraphics[width=0.5\linewidth,trim={0 0 0 0.12in}]{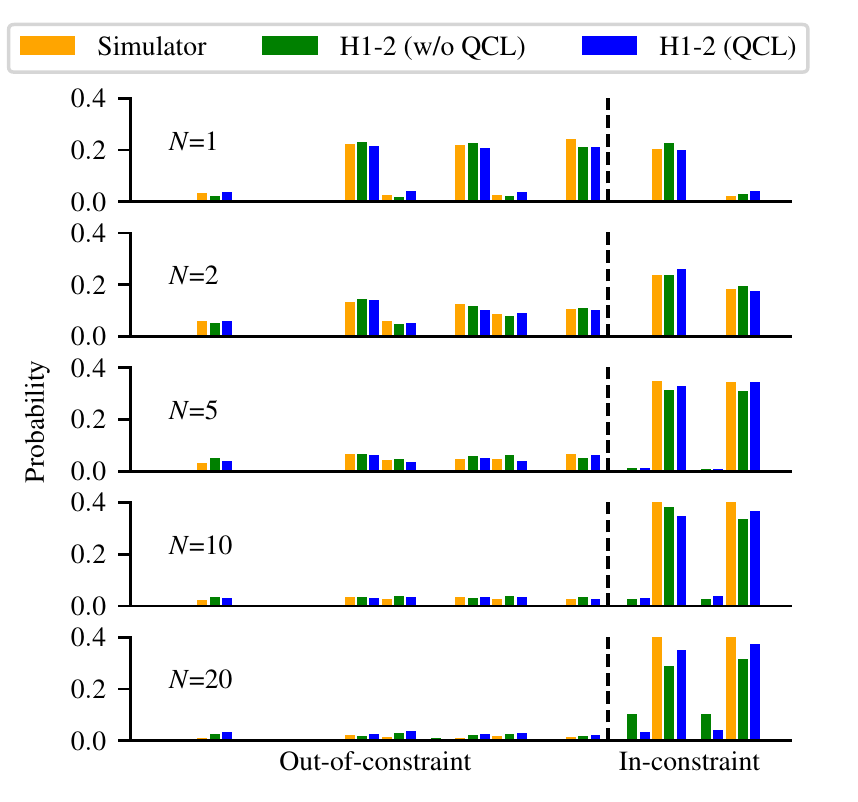}}  %
    \subfloat[]{\includegraphics[width=0.5\linewidth,trim={0 0 0 0.12in}]{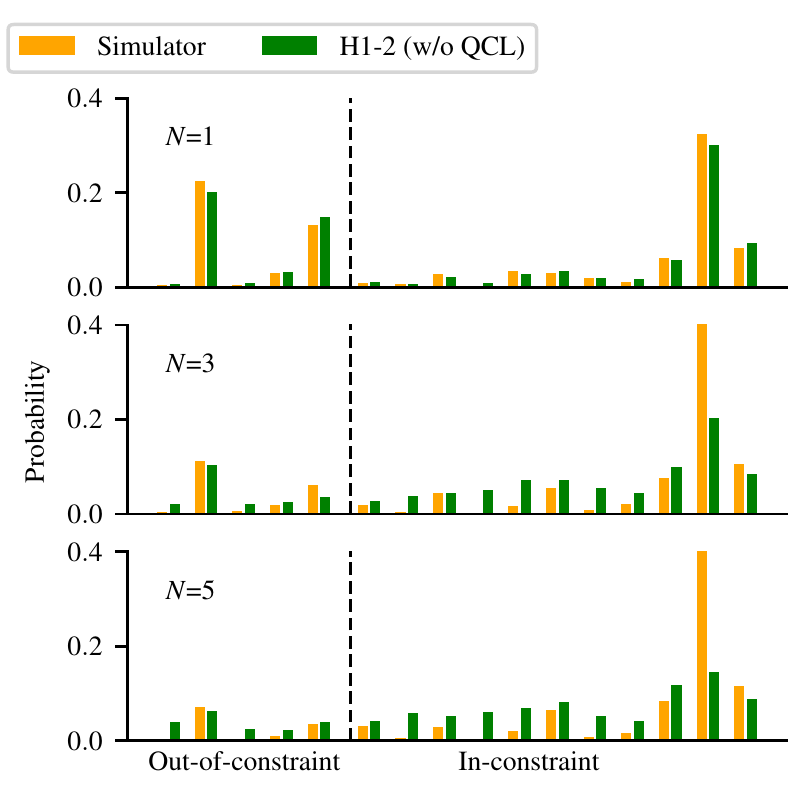}}  %
    \caption{
    \textbf{Effectiveness of constraint enforcement using QAOA with Zeno dynamics in simulation and hardware experiments.}
    Distribution of final measurement results obtained from QAOA applied to the equality- (a) and inequality-constrained (b) problem for different numbers of measurements ($N$). For the equality-constrained problem experiments were executed both with and without quantum conditional logic (QCL). Each column corresponds to a computational basis state (either in-constraint or out-of-constraint), and the columns are ordered by objective value (to the right is better). The circuits were executed on a classical simulator and on the H1-2 quantum device. There is strong agreement between the hardware results and results from noise-free simulation.}
    \label{fig:hardware_hist}
\end{figure*}

As a comparison, we also implement the coherent QFT (Figure~\ref{fig:qft_fourier_oracle}) on three qubits, resulting in seven qubits in total. After applying the oracle and measuring, all auxiliary qubits are reset to the ground state for the uncompute step. Figure~\ref{fig:hardware_delta_vs_N}a shows the in-constraint probability as a function of the number of projective measurements. Figure~\ref{fig:hardware_hist}a shows the distributions of measurement outcomes of QAOA for varying numbers of measurements ($N$), with the outcomes (computational basis states) ordered by the objective function value. For both implementations, the in-constraint probability improves with the number of measurements up to $N\approx 15$. For a higher number of measurements, the hardware noise arising from high circuit depth prevents further improvements in the in-constraint probability {$1-\delta$}.

While the QCL and non-QCL implementations both perform similarly, we do note a reduction in the number of  two-qubit gates and auxiliary qubits. For QCL and $N=15$, the two-qubit gate depth was 122 and the count was 123. Without QCL, for $N=15$, the two-qubit gate depth was 148 and the count was 165. The similar performance between QCL and non-QCL versions despite the difference in gate count may be due to the higher impact of measurement error on the QCL implementation.

The second portfolio optimization instance we consider has a cardinality (Hamming-weight) inequality constraint $\sum_{j=1}^4x_j \leq 2$. For this problem, it is necessary to utilize the coherent QFT, and thus QCL does not lead to a resource-requirement reduction. The QFT adder is used to compute $\sum_j x_j - 3$, which requires four qubits to accommodate the range. In addition, unlike the equality-constraint case, the inverse oracle is necessary for uncomputation. The system is in-constraint when the most-significant qubit, i.e., the sign bit, is a one. The circuit for the oracle is shown in Figure~\ref{fig:qft_fourier_oracle}. Similar to the previous run, we plot the in-constraint probability for varying numbers of measurements (Figure~\ref{fig:hardware_delta_vs_N}b), as well as, the measurement distributions obtained from QAOA (Figure~\ref{fig:hardware_hist}b). For $N = 3$, the two-qubit gate depth is 112 and the count is 186. Similarly to the experiments with the equality constraint, the in-constraint probability $1-\delta$ improves until $N=3$. For a higher number of measurements, the hardware noise prevents further improvements. 

Note that the performance deteriorates at a significantly lower $N$ for the inequality constraint problem than equality. This occurs even though the two-qubit circuit depth is lower for the inequality case and the two-qubit gate count is not significantly higher. Besides the inclusion of an additional qubit, one potential reason for this is that for the inequality constraint, only one of the auxiliary qubits is measured and then the inverse oracle is applied. This allows for errors to accumulate more and propagate to the rest of circuit. However, in the equality constraint case, after applying the oracle, all auxiliary qubits are measured and then reset to the ground state. In addition, the total gate count happens to be significantly higher for the inequality constraint case.

\section{Discussion}

In this work, we propose an approach for enforcing constraints in quantum optimization and demonstrate its effectiveness by applying it to constrained instances of portfolio optimization in simulation and on a trapped-ion quantum processor. Our technique has two major advantages: the ability to enforce a very general class of constraints and the simplicity of hyperparameter tuning. Two important downsides of our approach are the complexity of implementing the measurement and the possibility of the measurements resulting in trivial dynamics.

Implementing the oracle for a constraint in general requires quantum arithmetic and may lead to high gate count for more complex constraints. However, the asymptotic efficiency of our approach makes it viable for fault-tolerant quantum devices. Additionally, reductions in the cost of implementing quantum arithmetic, such as techniques utilizing quantum conditional logic, can further reduce the overhead of the proposed method.

Moreover, for noisy quantum devices, additional performance improvements can be obtained by leveraging advanced algorithm-specific error mitigation techniques such as the ones recently proposed for QAOA~\cite{Shaydulin2021,2204.05852}. Such techniques may help bridge the gap between the noisy near-term devices and the error correction likely required to execute circuits of sufficient depth to provide performance improvements over classical algorithms~\cite{takagi2022fundamental,farhi2020quantum,Sanders2020}.

As discussed in \alt{the \emph{Results} Section}{Section~\ref{sec:mixer_limitations}}, restricting the evolution to the Zeno subspace may result in trivial dynamics for certain mixers. Therefore an important consideration when applying the proposed technique is evaluating whether the particular choice of mixer has this behavior. As this effect would apply generally to all instances with a given class of constraints, the mixer only needs to be analyzed once for a class of problems.

\section{Methods}

\subsection{Preliminaries} \label{sec:preliminaries}

We begin by briefly introducing the relevant concepts and setting the notation. We undertake the task of minimizing an objective function $f$ defined on the Boolean cube, $\mathbb{B}^n$, over the set of feasible solutions $\F\subseteq\mathbb{B}^{n}$:
\begin{equation}
    \min_{\bx\in \F}f(\bx).
\end{equation}
We consider sets $\F$ of the form $\F=\{\bx\in\mathbb{B}^{n} \;\vert\; \Bar{g}_j(\bx)=0 \; \forall j\}$, where $\Bar{g}_j(\bx)$ is an oracle that returns $0$ if $\bx$ satisfies the $j$-th constraint and a value strictly greater-than $0$ otherwise. This general definition includes most commonly considered problems such as those with equality and inequality constraints. 

This constrained optimization problem can be solved by relaxing the constraints and introducing penalty terms as follows:
\begin{equation}
    \label{eq:obj_relaxed}
    \min_{\bx\in \mathbb{B}^n}f_{\text{penalty}} = \min_{\bx\in \mathbb{B}^n}f(\bx) + \sum_j \lambda_j \Bar{g}_j(\bx),
\end{equation}
where $\lambda_j \in \mathbb{R}^+$ are the penalty factors.

Specifically, for an equality constraint $g(\bx)=0$, the penalty function may be written as 
\[
    \Bar{g}(\bx) = \left[ g(\bx) \right]^2.
\]
On the other hand, an inequality constraint $g(\bx) \ge 0$ can be converted into an equivalent equality constraint $g(\bx) - \hat{s} = 0$ by introducing a \emph{slack variable} $\hat{s} \in [0, g_{\max}]$, where $g_{\max} = \max_{\bx \in \F} g(\bx)$.
If we assume $g(\bx)$ can be discretized with a spacing of $\Delta_g$, then $\hat{s}$ can be implemented using $n_{\text{slack}} = \lceil \log_2 (g_{\max} / \Delta_g) \rceil$ binary variables $\boldsymbol{s} = (s_1, \ldots, s_{n_{\text{slack}}})^{\mathsf{T}}$, and the resultant equality constraint is $g(\bx) - \Delta_g \sum_{j} 2^{j-1}s_j = 0$.
Therefore the penalty function for an inequality constraint can be written as 
\[
    \Bar{g}(\bx;\boldsymbol{s}) = \left[g(\bx) - \Delta_g \sum_{j=1}^{n_{\text{slack}}} 2^{j-1}s_j\right]^2.
\]

The magnitudes of the penalty factors $\lambda_j$ control how much the constraint violations are penalized. Intuitively, a higher value of $\lambda_j$ should lead to a higher in-constraint probability. However, in practice, the relationship between the penalty factor, the in-constraint probability and the solution quality may be non-monotonic. This makes choosing $\lambda_j$ harder. We discuss the difficulty of tuning the penalty factors in \alt{the \emph{Results} Section}{Section~\ref{sec:penalty_factor_tune}}.

\subsubsection{Quantum algorithms for approximate optimization} \label{sec:preliminaries_qopt}

In this work, we focus on the class of quantum optimization algorithms that use a parameterized quantum evolution to prepare a state, such that the corresponding measurement outcomes contain a high-quality, valid solution to the original optimization problem with high probability. This parameterized state, a restatement of Equation~\eqref{eq:results_ansatz_vqe}, is prepared by applying a parameterized evolution $U(\boldsymbol{\theta})$ to some initial state $\ket{s}$:
\begin{equation}
    \label{eq:ansatz_vqe}
    \ket{\psi(\boldsymbol{\theta})} = U(\boldsymbol{\theta})\ket{s} = \prod_{j=1}^{m} e^{-i\theta_jH_j} \ket{s},
\end{equation}
where $H_j$ is some Hamiltonian, e.g., a tensor product of single-qubit Pauli operators.

Let $C=\sum_{\bm{x}\in\mathbb{B}^n}f(\bm{x})\ketbra{\bm{x}}$ be the operator encoding the objective function $f$ on qubits and  $C_{\text{penalty}}=\sum_{\bm{x}\in\mathbb{B}^n}f_{\text{penalty}}(\bm{x})\ketbra{\bm{x}}$ be the operator encoding the relaxed objective function \eqref{eq:obj_relaxed}. The figures of merit used to evaluate the quality of a parameter $\bm{\theta}^*$ obtained by algorithms that employ parameterized circuit \eqref{eq:ansatz_vqe} are approximation ratios, defined as follows:
\begin{equation}\label{eq:approx_ratio_vqe}
    r = \frac{\bra{\psi(\boldsymbol{\theta^*})}C_{\F}\ket{\psi(\boldsymbol{\theta^*})} - f^{\max}}{f^{\min} - f^{\max}}
\end{equation}
and
\begin{equation}\label{eq:approx_ratio_penalty_vqe}
    r_{\text{penalty}} = \frac{\bra{\psi(\boldsymbol{\theta^*})}C_{\text{penalty}}\ket{\psi(\boldsymbol{\theta^*})} - f_{\text{penalty}}^{\max}}{f_{\text{penalty}}^{\min} - f_{\text{penalty}}^{\max}},
\end{equation}
where $C_{\F}=\sum_{\bm{x}\in\F}f(\bm{x})\ketbra{\bm{x}}$, $f^{\min}=\min_{\bx\in \F}f(\bx)$, $f^{\max}=\max_{\bx\in \F}f(\bx)$, $f_{\text{penalty}}^{\min}=\min_{\bx\in \mathbb{B}^{n}}f_{\text{penalty}}(\bx)$, and $f_{\text{penalty}}^{\max}=\max_{\bx\in \mathbb{B}^{n}}f_{\text{penalty}}(\bx)$.

This class of algorithms includes QAOA~\cite{Hogg2000,farhi2014quantum, Sanders2020} and its generalization, the quantum alternating operator ansatz algorithm~\cite{Hadfield_2019}. In both algorithms, the parameterized quantum evolution is performed by applying pairs of alternating operators:
\begin{equation}
    \label{eq:ansatz_qaoa}
    \ket{\psi(\boldsymbol{\beta},\boldsymbol{\gamma})} = \prod_{j=1}^{p}\left[U_B(\beta_j)U_C(\gamma_j)\right]\ket{s},
\end{equation}
where $U_C(\gamma_j)=e^{-i\gamma_j C}$ is the phase operator, and $U_B(\beta_j)$ is the mixing operator. In the special case of QAOA, the initial state $\ket{s}$ is the uniform superposition over all computational basis states and the mixing operator $U_B$ is set to be $U_B(\beta_j)=e^{-i\beta_j B}$, where $B=\sum_k\xgate_k$ is a sum of single-qubit Pauli-$\xgate$ operators. In quantum alternating operator ansatz, $U_B$ and $\ket{s}$ are allowed to be arbitrary, and are typically set such that the resulting state $\ket{\psi(\boldsymbol{\beta},\boldsymbol{\gamma})}$ preserves the constraints, in the sense that every measurement outcome $\bm{x}$ belongs to $\F$. In this paper, we consider QAOA with an arbitrary mixing Hamiltonian $B$, defined in Ref.~\cite{Hadfield_2019} as Hamiltonian-based QAOA.  In all other sections of this paper, unless it is specified otherwise, the acronym QAOA is used to denote this version of the algorithm. 

In addition to QAOA, we consider the layer variational quantum eigensolver (L-VQE)~\cite{Liu2022}, which is a version of VQE with the hardware-efficient layered parameterized circuit tailored towards optimization problems. L-VQE uses the parameterized circuit of the form
\begin{equation}\label{eq:ansatz_lvqe}
    \prod_{j=1}^p \big[U_{\text{NN}}(\boldsymbol{\theta}_j)\big]V(\boldsymbol{\theta}_0)\ket{0},
\end{equation}
where $U_{\text{NN}}$ consists of nearest-neighbor $\cnotgate$'s and single-qubit $\rygate$'s, and $V$ is a layer of single-qubit $\rygate$'s. The reader is referred to Ref.~\cite{Liu2022} for the precise definition of the circuit. While the circuit includes non-parameterized $\cnotgate$'s, it is easy to write it equivalently in the form of Equation~\eqref{eq:ansatz_vqe} by pushing $\rygate$ through the control of the $\cnotgate$
and noting that $\rygate(\theta) = e^{-i\frac{\theta}{2}\ygate}$ and $\cnotgate_{1,2}\rygate_2(\theta)\cnotgate_{1,2} = e^{-i\frac{\theta}{2}\zgate_1\ygate_2}$. Here, $\ygate_j$ and $\zgate_j$ denote a single-qubit Pauli-$\ygate$ and Pauli-$\zgate$, respectively,  acting on the $j$-th qubit.

\subsubsection{Quantum Zeno dynamics}%
The quantum Zeno effect (QZE)~\cite{zenosparadox, presilla1996measurement} is named after Zeno's paradox~\cite{sep-paradox-zeno}, which regards the continuous observation of a moving arrow. Zeno's paradox states that an arrow cannot move if no time has elapsed since the point it was last observed. If the time difference between observations is $\Delta t$, continuous observation occurs in the limit of $\Delta t \xrightarrow{} 0$. Under continuous observation, no time elapses between observations, and during each observation the arrow is not moving; thus, no overall movement is possible. The analog in quantum mechanics is a consequence of the Schr\"{o}dinger equation. We first introduce a simpler one-dimensional version, in which the quantum state is restricted from evolving due to repeated measurements, and then present a more general case in which the dynamics of the system are restricted to a particular  subspace, called a Zeno subspace.

Suppose a time-dependent quantum state is evolved in a finite-dimensional Hilbert space $\calH$ from some initial state {$\ket{\psi_{0}}$} under the action of some Hamiltonian $H$ for time $t$. Define a projective measurement $\calP$ given by a pair of complement projections $P=\ketbra{\psi_0}{\psi_0}$ and $Q = \idgate-P$, which acts on a density operator $\rho$ as
\[
    \calP \rho = P \rho P + Q \rho Q.
\]
If we carry out $N$ repeated projective measurements $\calP$ at a time interval of $t/N$, then the probability that the system remains in the initial state is
\begin{align*}
    p(t) & = \lVert Pe^{-iHt/N}\ket{\psi_0}\rVert_{2}^{2N} \\
    & = \left[\lvert\bra{\psi_0}e^{-iHt/N}\ket{\psi_0}\rvert^{2}\right]^{N} \\
    & = \left[1-(t/N\tau_Z)^2\right]^{N} + O(N^{-2}) \xrightarrow{ \;N \rightarrow \infty \;}  1,
\end{align*}
where $\tau_Z^{-2} = \bra{\psi_0}H^2\ket{\psi_0} - \bra{\psi_0}H\ket{\psi_0}^2$ is called the Zeno time and quantifies how often the measurements need to be taken. 
As the frequency at which the measurements are performed increases without bound, the probability of remaining in the initial state approaches one.

Quantum Zeno dynamics (QZD) \cite{facchi2000quantum, Facchi_2002, Facchi_2008, Burgarth_2020} considers the more general case where the evolution of the state is constrained to a subspace of dimension greater than one. Thus the projective measurement $\calP$ can contain multiple projections with ranks all greater than one. Specifically, a restatement of Equation~\eqref{eq:results_measurement},
\begin{equation}
    \label{eq:measurement}
    \calP \rho = \sum_{j=1}^k P_j \rho P_j,
\end{equation}
where $\sum_{j=1}^k P_j = \idgate$, and $P_j$ is a projection onto some subspace $\calH_{j} = P_j\calH$ of dimensionality $\Tr(P_j) \ge 1$.
Informally, QZD states that if the evolution starts in $\calH_j$ and the measurement $\calP$ is performed sufficiently often, then the system will remain in $\calH_j$ with high probability. 

Consider an initial state $\rho_0$, after $N$ projective measurements by $\calP$, the state of the system is given by
\begin{equation}
    \rho(t) = \U(t)\rho_0\U(t)^\dagger,
\end{equation} 
where $\U(t) = \left(\mathcal{P} e^{-iHt/N}\right)^N$ and $p(t) = \Tr(P_j\rho(t))$ is the probability of the system remaining in $\calH_j$ after evolving for time $t$.
Note that 
\begin{align}
    \U(t) & = \left(\mathcal{P} e^{-iHt/N}\right)^N \label{eq:zeno_first} \nonumber \\
    & = \left(\mathcal{P}[\idgate-iHt/N + O( N^{-2})] \right)^N \nonumber \\
    & = \left(\idgate-i\mathcal{P}H t/N + O( N^{-2})\right)^N \nonumber\\
    &=\left(\idgate-i\mathcal{P}H t/N\right)^N + O( N^{-1}) \\
    & \xrightarrow{\; N \rightarrow \infty \;} e^{-i\mathcal{P}Ht}\mathcal{P},  \label{eq:zeno_last}
\end{align}
and the dynamics of the system are governed by $H_Z = \calP H$, called the Zeno Hamiltonian. Moreover, as $N\rightarrow\infty$, transitions between different subspaces $\{\calH_1, \dots, \calH_k\}$ of $\calH$ are suppressed. This implies if $\rho_{0} = P_j\rho_0P_j$ for some $j \in [k]:=\{1, \dots, k\}$, then in the limit of $N \xrightarrow[]{} \infty$, called the Zeno limit, it follows that $p(t)\rightarrow 1$, and thus the state will remain in $\calH_{j}$ throughout the evolution.
For a more detailed discussion the reader is referred to Refs.~\cite{Facchi_2008, Burgarth_2020}.

QZE has many applications in algorithms and error mitigation. Childs et al. \cite{childs2002} propose a version of Grover's search based on QZD that utilizes frequent measurements instead of slow adiabatic evolution. This alternative approach to slow evolution was also observed in Ref.~\cite{aharonov2003}. Somma et al. \cite{somma2007, Somma_2008} develop a quantum-enhanced version of the simulated annealing algorithm. Their approach makes use of QZD to ensure that the evolution remains in the instantaneous quantum Gibbs state for varying temperature. Boixo et al. \cite{Boixo2009} show that for Grover's algorithm and simulated annealing based on QZD, one could use frequent randomized evolutions instead of measurements (the randomization method). The randomization method has also been used to implement algorithms for quantum linear systems \cite{Suba2019QuantumAF, Lin2020optimalpolynomial}. Finally, dynamical decoupling, also called bang-bang decoupling~\cite{Viola1998}, is a popular error-mitigation technique that uses QZE to suppress decoherence \cite{Facchi_2004, Burgarth_2020, Halimeh2021, Halimeh2021_2, Halimeh2022, Halimeh2022_2}.

\subsection{Proof of Theorem 1} \label{sec:proof_of_main_thm}
In this Section we derive our main result, Theorem \ref{thm:main_theorem}, for the number of measurements required to maintain a constant success probability. We start by deriving the required lemmas. 

\begin{manuallemma}{1}
\label{lem:scaling_lemma}
Let $H$ be a Hermitian  matrix. Then 
\begin{align*}
    \min_{P, \ket{\psi} \in \Im(P)} & \norm{P e^{-i \theta H} \ket{\psi}}_{2}^2 = \cos^2\left(\frac{\xi_{\max} - \xi_{\min}}{2} \theta \right) \\ 
    &\forall \theta \in \mathbb{R}, \abs{\theta} \le \frac{\pi}{\xi_{\max} - \xi_{\min}},
\end{align*}
where $P$ is an orthogonal projector and $\xi_{\max}$ and $\xi_{\min}$ are the largest and smallest eigenvalues of $H$.
\end{manuallemma}

\begin{proof}
Suppose $H$ has the following eigendecomposition
\[
    H = \sum_{k=1}^d \xi_k Q_k,
\]
where $\xi_k$ are the unique eigenvalues of $H$ (including $0$ if $H$ is not full rank) and $\{Q_k\}_{k=1}^d$ is the complete set of projectors onto the corresponding eigenspaces. Therefore
\begin{align}
\label{eqn:one_step_in_constraint_tight_proof}
p(\theta) =& \norm{P e^{-i \theta H} \ket{\psi}}_{2}^2 \nonumber \\
    \ge& \norm{\ketbra{\psi}{\psi} e^{-i \theta H} \ket{\psi}}_{2}^2 \nonumber \\
    =& \abs{\bra{\psi}e^{-i \theta H}\ket{\psi}}^2 \nonumber \\
    =& \sum_{j,k=1}^d e^{i \theta (\xi_j - \xi_k)} \bra{\psi}Q_j\ket{\psi} \bra{\psi}Q_k\ket{\psi} \nonumber \\
    =& \sum_{j,k=1}^d \cos(\theta (\xi_j - \xi_k)) \bra{\psi}Q_j\ket{\psi} \bra{\psi}Q_k\ket{\psi}  \nonumber \\
    =& \sum_{j,k=1}^d c_{jk} x_j x_k,
\end{align}
where $c_{jk} = \cos(\theta (\xi_j - \xi_k))$,  $x_j = \bra{\psi}Q_j\ket{\psi} \geq 0$. Note that the second to the last equality follows from 
\begin{align*}
e^{i \theta (\xi_j - \xi_k)}x_{j}x_{k} + e^{i \theta (\xi_k - \xi_j)}x_{k}x_{j} %
= \cos(\theta(\xi_j - \xi_k))x_{j}x_{k} + \cos(\theta(\xi_k - \xi_j))x_{k}x_{j}.
\end{align*}Let $C$ be the matrix with elements $c_{jk}$ at the $j$-th row and $k$-th column. Then using simple trigonometric identities, it can be shown that
\begin{align}
    C &= \bm{v}(\theta)\bm{v}(\theta)^{\mathsf{T}} + \bm{v}\left(\frac{\pi}{2} - \theta\right)\bm{v}\left(\frac{\pi}{2} - \theta\right)^{\mathsf{T}}
\end{align}
where 
\begin{align}
    &\bm{v}(\theta) = ( \cos(\xi_1\theta), \dots, \cos(\xi_d\theta))^{\mathsf{T}}.
\end{align}
Since $C$ is the sum of positive semi-definite matrices, it too is positive semi-definite. 

Therefore, minimizing $p(\theta)$ is equivalent to solving the following convex constrained minimization problem
\begin{equation}
\label{eq:minimization_term}
    \min_{\boldsymbol{x} \in \mathcal{S}} \boldsymbol{x}^{\mathsf{T}} C \boldsymbol{x}, \text{where} \ \mathcal{S} := \{ \bm{x} \in \mathbb{R}_{+}^{d}~|~\lVert \bm{x} \rVert_1 = 1\},
\end{equation}
$\boldsymbol{x} = (x_1, \ldots, x_d)^{\mathsf{T}}$ and thus a sufficient condition \cite[Theorem 2.2.5]{nesterov2003introductory} for $\bm{x}^{\star}$ to be the optimum is

\begin{equation}
\label{eq:optimality_condition}
    {\bm{x}^\star}^{\mathsf{T}}C(\bm{x} - \bm{x}^{\star}) \geq 0, ~ \forall \bm{x} \in \mathcal{S}
\end{equation} 
Consider the following trial solution
\begin{align}
\label{eq:kkt_trial_solution}
    &x^\star_{\min} = x^\star_{\max} = \frac{1}{2}, \nonumber \\
    &x^\star_j = 0 \quad \forall \ j \not \in \{\min, \max\}.
\end{align}
We have that $\forall \bm{x} \in \mathcal{S}$
\begin{align*}
& 2{\bm{x}^\star}^{\mathsf{T}}C(\bm{x} - \bm{x}^{\star}) \\
&= (1+c_{\max,\min})(x_{\max} +x_{\min} - 1)
+ \sum_{j \notin \{\min, \max\}}x_j(c_{\max, j} +c_{\min, j}) \\
&= (1 - x_{\max} - x_{\min})\Bigg[\sum_{j \notin \{\min, \max\}}(c_{\max, j} +c_{ \min, j}) - (1+c_{\max, \min})\Bigg]
\end{align*}
Also for $\abs{\theta} \le \frac{\pi}{\xi_{\max} - \xi_{\min}}$, we have $c_{j,k} \geq c_{\max, \min}$, and thus
\begin{align*}
    &1 + c_{\min,\max} = 2 \cos^2\left(\frac{\xi_{\max} - \xi_{\min}}{2} \theta \right) \\
    \le& 2 \cos\left(\frac{\xi_{\max} - \xi_{\min}}{2} \theta \right) \cos\left(\frac{\xi_{\max} + \xi_{\min} - 2\xi_{j}}{2} \theta \right)\\
    =& c_{\max, j} + c_{\min, j}.
\end{align*}
Combining the above results, we obtain
that $2{\bm{x}^\star}^{\mathsf{T}}C(\bm{x} - \bm{x}^{\star})\geq 0$. Thus our choice is optimal.

After, plugging in the optimal choice and noting that all steps are equalities in \eqref{eqn:one_step_in_constraint_tight_proof} when $P=\ketbra{\psi}$, we obtain:
\begin{align*}
    \min_{P, \ket{\psi} \in \Im(P)} \norm{P e^{-i \theta H} \ket{\psi}}_{2}^2 
    = {\boldsymbol{x}^\star}^{\mathsf{T}} C \boldsymbol{x}^\star
    = \cos^2\left(\frac{\xi_{\max} - \xi_{\min}}{2} \theta \right).
\end{align*}
Additionally, the result implies that minimization occurs when 
\begin{align}
\label{eqn:worst_case_state}
\ket{\psi} = \ket{\pm_{H}}:=\frac{1}{\sqrt{2}}\ket{\xi_{\max}} \pm\frac{1}{\sqrt{2}}\ket{\xi_{\min}}
\end{align} for any $\ket{\xi_{\max}} \in \Im(Q_{\max})$ and $\ket{\xi_{\min}} \in \Im(Q_{\min})$. 
\end{proof}

Note as observed in the proof of Lemma~\ref{lem:scaling_lemma}, the lower bound on the in-constraint probability bound is saturated when the initial state is chosen to be either $\ket{+_{H}}$ or $\ket{-_{H}}$ in Equation \eqref{eqn:worst_case_state}, and $P$ is the projector onto the chosen initial state.

\begin{manuallemma}{2}
\label{cor:prob_lower_bound_tight}
Let $H$ be a Hermitian matrix. Then
\begin{align*}
    \min_{P, \ket{\psi} \in \Im(P)} &\norm{P \left(\calP e^{-i \frac{\theta}{N} H} \right)^{N} \ket{\psi}}_{2}^{2}
    = \frac{1}{2} + \frac{1}{2} \left[2 \, p^*\left(\frac{\theta}{N}\right)-1\right]^{N}, \\
    &\forall \theta \in \mathbb{R}, \ \abs{\theta} \le \frac{\pi N}{\xi_{\max} - \xi_{\min}},
\end{align*}
where $\calP$ is a projective measurement as defined in Equation~\eqref{eq:measurement} with projectors $P$ and $\idgate - P$,
\[
    p^*\left(\frac{\theta}{N}\right) = \cos^{2}\left(\frac{\xi_{\max} - \xi_{\min}}{2N} \theta \right),
\]
and $\xi_{\max}$ and $\xi_{\min}$ are the largest and smallest eigenvalues of $H$.
\end{manuallemma}
\begin{proof}
Consider a fixed $\theta$ and some $N$ that satisfies the hypothesis. The stochastic process formed by random variables indicating whether the system is in $\Im(P)$ or its complement after each evolution segment $\mathcal{P}e^{-i\frac{\theta}{N}H}$ form a two-state Markov chain. According to Lemma \ref{lem:scaling_lemma}, the probability of remaining in a state on the chain at any point in time is at least 
\begin{equation}
    \label{eqn:cos_prob}
    p^*\left(\frac{\theta}{N}\right) := \cos^{2}\left(\frac{\xi_{\max} - \xi_{\min}}{2N} \theta \right),
\end{equation}
and this minimum probability is attained at each segment when $\ket{\psi}$ is \eqref{eqn:worst_case_state} and $P=\ketbra{\psi}$. Because, in this case, the evolution lies in the two-dimensional space spanned by $\ket{\pm_{H}}$, the result is a Markov chain with transition matrix
\begin{align}
A(k) = \bar{A} = \begin{pmatrix}
p^* & 1 - p^* \\
1 - p^* & p^*
\end{pmatrix}, \forall k \in [N],
\end{align}
and $\forall k > N, A(k) = \idgate$. 

Therefore the probability of the state remaining in $\Im(P)$ after $N$ steps of the chain is $\bar{A}^{N}_{0,0}$, or the first diagonal element of the matrix $\bar{A}$ after raising it to the $N$-th power. Applying diagonalization on $\bar{A}$, we obtain
\begin{align}
    \bar{A}^{N}_{0,0} = \frac{1 + (2p^*-1)^{N}}{2}.
\end{align}
\end{proof}

We now proceed to derive Theorem \ref{thm:main_theorem} using the above lemmas.

\begin{proof}[Proof of Theorem \ref{thm:main_theorem}]
 For all $ \theta \in \mathbb{R}$, such that
 \begin{align}
    \label{eqn:theta_lower_bound}
     \abs{\theta} < \frac{N}{\xi_{\max} - \xi_{\min}},
\end{align}
it follows that
\begin{align}
    \cos^{2}\left(\frac{\xi_{\max} - \xi_{\min}}{2N}\theta\right)
    \geq &\left(1 - \frac{1}{2}\left[\frac{\theta(\xi_{\max} - \xi_{\min})}{2N}\right]^2\right)^{2}\nonumber\\
    \geq &1 - \frac{\left[\theta(\xi_{\max} - \xi_{\min})\right]^2}{4N^2}\nonumber.
\end{align}
If we combine this result with  Lemma~\ref{cor:prob_lower_bound_tight}, then we obtain
\begin{align}
    \frac{1}{2} + \frac{1}{2} \left[2 \, p^*\left(\frac{\theta}{N}\right)-1\right]^{N}\nonumber
    &\geq \frac{1}{2} + \frac{1}{2} \left(1 -\frac{\left[\theta(\xi_{\max} - \xi_{\min})\right]^2}{2N^2}\right)^{N}\nonumber\\
    &\geq\frac{1}{2} + \frac{1}{2}\exp(-\frac{\left[\theta(\xi_{\max} - \xi_{\min})\right]^2}{2N})
\end{align}
To lower bound this by ${1 - \delta}$, we can choose $N$ as stated in Theorem \ref{thm:main_theorem}. Note that to ensure Equation~\eqref{eqn:theta_lower_bound} we must have
\begin{align}
    \frac{\left[\theta (\xi_{\max} - \xi_{\min})\right]^2}{N} < N,
\end{align}
and thus
\begin{align}
    \frac{1}{2} + \frac{1}{2}\exp(-\frac{\left[\theta (\xi_{\max} - \xi_{\min})\right]^2}{2N})
    > \frac{1}{2} + \frac{1}{2}\exp(-\frac{N}{2}).
\end{align} At the minimum of value of $N$, we have
\begin{align}
\frac{1}{2} + \frac{1}{2}\exp(-\frac{1}{2}) \lesssim 0.81.
\end{align}

\end{proof}

\subsection{Proof of Corollary 1}
\label{sec:proof_of_cor1}
\begin{proof} {For simplicity, consider a single block of size $m$:
\begin{align}
    \U_Z(\bm{\theta}) = \left[\prod_{j=1}^{m}e^{-i(\theta_{j}/N)H_j}\right]^{N}.
\end{align}
First, suppose that the elements of $\{H_j\}_{j=1}^{m}$ do not all pairwise commute. Then, according to \cite[Proposition~9]{childs2021theory}:
\begin{align}
    \norm{ \prod_{j=1}^{m}e^{-i(\theta_{j}/N)H_j} - e^{-i\sum_{j=1}^{m}(\theta_{j}/N)H_{j}} }_{2} \leq \frac{1}{2N^2}\sum_{j=1}^{m} \norm{ \left[ \sum_{j'=j+1}^{m}\theta_{j'}H_{j'}, \theta_{j}H_j \right] }_{2}
\end{align}
This implies that
\begin{align}
    \norm{ \U_Z(\bm{\theta}) - \left[\mathcal{P}e^{-i\sum_{j=1}^{m}(\theta_j/N)H_j}\right]^{N} }_{2}
    &\leq \frac{1}{2N}\sum_{j=1}^{m} \norm{ \left[ \sum_{j'=j+1}^{m}\theta_{j'}H_{j'}, \theta_{j}H_j \right] }_{2}\nonumber\\
    &\leq \frac{\left[\sum_{j=1}^{m}\lvert\theta_j\rvert\right]^2\max_{j} \norm{ H_{j} }_{2}^2}{N}.
\end{align}
Then
\begin{align}
    \norm{ P_{\mathcal{G}}\U_Z(\bm{\theta})\ket{\psi}}_{2}^{2} \leq 
    \left( \norm{ P_{\mathcal{G}} \left[\mathcal{P}e^{-i\sum_{j=1}^{m}(\theta_j/N)H_j} \right]^{N}\ket{\psi}}_{2}
     + \frac{\left[\sum_{j=1}^{m}\lvert\theta_j\rvert\right]^2\max_{j}\norm{ H_{j} }_{2}^2}{N}\right)^2.
\end{align}
If we choose 
\begin{align}
N = \left\lceil\frac{4\left[\sum_{j=1}^{m}\lvert\theta_j\rvert\right]^2\max_{j}\norm{ H_{j}}_{2}^{2}}{\ln{\left(1-\delta\right)^{-2\alpha}}}\right\rceil,
\end{align} 
then for $\alpha \leq 1$, Theorem \ref{thm:main_theorem} with Remark \ref{remark_1} implies that the out-of-constraint probability is at most 
\begin{align}
    \norm{ P_{\mathcal{G}}\U_Z(\bm{\theta})\ket{\psi}}_{2}^{2} 
    &\leq \frac{\delta}{2} + \alpha\frac{\sqrt{\delta}}{2}\ln{\left(1-\delta\right)^{-2}} + \frac{\alpha^2}{16}\ln^{2}{\left(1-\delta\right)}^{-2}\\
    &\leq \frac{\delta}{2} + \frac{\delta}{2}\left[\alpha + \frac{\alpha^2}{8}\right],
\end{align}
where $\delta \leq 0.19$. If $\alpha = 0.89$, then
\begin{align}
    \norm{ P_{\mathcal{G}}\U_Z(\bm{\theta})\ket{\psi}}_{2}^{2} < \delta.
\end{align}
}
{To compensate for the decay of the success probability after $L$ blocks, each $N_k$ must be multiplied by $L$.}

Lastly, for the asymptotic dynamics, from Equation~\eqref{eq:zeno_first}-\eqref{eq:zeno_last} we get
\begin{align}
    \U_Z(\bm{\theta}) 
    & = \left[\calP \prod_{j=1}^m\left(\idgate -i(\theta_j/N)H_j + O(N^{-2})\right)\right]^N \nonumber\\ 
    & = \left[\calP \left(\idgate -i\sum_{j=1}^m (\theta_j/N)H_j + O(N^{-2})\right)\right]^N \\
    & \xrightarrow{\; N \rightarrow \infty \;} e^{-i \sum_{j=1}^m\calP H_j\theta_j}\calP = e^{-i \calP\bm{H}\cdot\bm{\theta}}\calP.
\end{align}
Thereby the dynamics are described by the Zeno Hamiltonian $\bm{H_Z} = \calP\bm{H}$, where $\calP$ acts element-wise on the vector $\bm{H} = (H_1, \dots , H_m)^{\mathsf{T}}$. The limiting dynamics of $L$ blocks is the product of these limits.

If the elements of $\{H_j\}_{j=1}^{m}$ pairwise commute, then there is no Trotter error, and $\alpha=1$ without the need to halve $\delta$. The limiting dynamics follows trivially as well.
\end{proof}

\subsection{Realizing oracles for combinatorial constraints} \label{sec:implement_oracle}
In this Section, we review the constructions of quantum oracles for implementing polynomial inequality and equality constraints. We use the constructions provided in this Section in the experiments on a trapped-ion quantum computer described in \alt{the \emph{Results} Section}{Section~\ref{sec:hardware_experiments}}. Since any function on the Boolean cube can be expressed as a polynomial it suffices to only demonstrate constructions for polynomial constraints \cite{odonnell_2014}. In addition, since we are considering problems in $\mathsf{NPO}$ we can assume the existence of a polyomially-sized classical circuit for evaluating any constraints to sufficient precision. Given that all classical basis gates can be represented as polynomials, we can represent our constraint as the composition of polynomially many polynomial functions. Of course, one could also directly implement the classical circuit in a reversible fashion on a quantum device efficiently. For the remainder of this Section, we consider a polynomial function $g$:
\begin{equation}
\label{eqn:binary_poly}
g(\bm{b}) = \sum_{k=1}^{K} d_k\prod_{l \in S_k}b_l,
\end{equation}
where $S_k \subseteq [n]$ and $d_k \in \mathbb{R}$. In addition for $S_k = \varnothing$, $\prod_{l \in S_k}b_l := 1$. 

Without loss of generality we can assume that equality constraints are of the form $g(\bm{b}) = 0$ and inequality constraints are of the form $g(\bm{b}) \geq 0$. We assume that there exists an oracle that computes the value of $g(\bm{b})$ into a quantum register (constructions of such oracles are briefly reviewed in \alt{the \emph{Methods} Section}{Sections~\ref{sec:classical_reversible_arithmetic}~and~\ref{sec:quantum_fourier_arithmetic}}). For an equality constraint, we implement the constraint-enforcing measurement by simply measuring the entire register. A projection onto the in-constraint subspace implies that we have observed a $0$ in the register. For an inequality constraint, we measure the qubit corresponding to the sign, a $0$ corresponds to a successful projection, and apply the inverse of the oracle post measurement. 

While the above procedure works in general, there are further optimizations that can be made by utilizing quantum conditional logic (QCL). We give an example of such an optimization in \alt{the \emph{Results} Section}{Section~\ref{sec:hardware_experiments}}. Further optimizations are possible for double-sided inequalities of the form $0\leq g(\bm{b}) < a$, where $a$ is a power of $2$. To implement the measurement corresponding to this double-sided inequality, we only need to measure higher-order bits. Since the results of these high-order bits are now classical, we can replace the part of the inverse-oracle circuit controlled on these bits with classically-conditioned single-qubit gates. Lastly, because all constraint-preserving measurements can be implemented separately and thus auxiliary qubits can be reused, the required number of auxiliary qubits to implement all constraint-preserving measurements is equal to the maximum amount of auxiliary qubits required by any oracle call.

In the subsections that follow, we present efficient constructions of oracles that can be used to implement polynomial functions. Both of these use techniques that have been presented in prior work. Here we include a brief review for completeness and present the resource analysis for our setting. 

\subsubsection{Review of classical reversible arithmetic circuits} \label{sec:classical_reversible_arithmetic}
The design of reversible versions of classical arithmetic circuits has been extensively explored and highly optimized constructions are available~\cite{haner2016, Hner2018OptimizingQC, Haner_2018_floating}. Such constructions allow one to implement unitary operations for performing arithmetic on quantum registers.
Consider fixed-point arithmetic of $m$ bits including digits both before and after the decimal point. Suppose polynomial $g$ has $K$ terms. For each coefficient $d_k$, we require an $n$-qubit controlled $m$-bit adder. A controlled $m$-bit adder can be implemented with $O(m)$ $\tgate$ gates \cite{Gidney_2018}. Since a multi-controlled Toffoli can be implemented with a $\tgate$ of $O(n)$ \cite{Jones_2013, maslov2016} and thus the overall multi-controlled adder can be implemented with a $\tgate$ count of $O(n+m)$. The $\tgate$ count for implementing $g$ is $O(K(n+m))$.

\subsubsection{Review of quantum Fourier arithmetic} \label{sec:quantum_fourier_arithmetic}
For smaller quantum devices, a more resource efficient approach is to switch to the Fourier basis using the quantum Fourier transform (QFT) and perform the arithmetic in the Fourier basis. This approach has worse asymptotic complexity in terms of $\tgate$-gate counts, but requires fewer qubits and \cnotgate~gates. We use this approach in the hardware experiments discussed in \alt{the \emph{Results} Section}{Section~\ref{sec:hardware_experiments}}. The discussion in this Section is based on Ref.~\cite{Gilliam_2021}, though the idea of using the QFT for quantum arithmetic is well-known, see e.g. \cite{draper2000, Ruiz_Perez_2017, ahin_2020}. 

For $s \in [2^m]$, the QFT on $\mathbb{Z}_{2^m}$  is defined as follows:
\begin{equation}
\label{eqn:qft}
    \text{QFT}_{2^m}: \ket{s} \mapsto \sum_{k \in [2^m]}e^{-i2\pi k s/{2^m}}\ket{k}.
\end{equation}
It can be shown \cite{nielsen2010quantum} that the right-hand side of \eqref{eqn:qft} is a product state and can be expressed in the following form:
\begin{equation}
     \bigotimes_{k=1}^{m} \frac{\ket{0} + e^{-i\pi \frac{ s}{2^{m-k}}}\ket{1}}{\sqrt{2}} = F_m\left(\frac{s}{2^m}\right)\ket{+}^{\otimes m},
\end{equation}
where
\begin{equation}
    F_m(\theta) :=\bigotimes_{k=1}^{m} R(\pi 2^k\theta)
\end{equation}
implements the desired operation. In addition, $R(\alpha)$ denotes the phase gate $\ketbra{0} + e^{i\alpha}\ketbra{1}$. The angle $\theta$ is restricted to $[-\frac{1}{2}, \frac{1}{2})$ to avoid overflow and allow for representing negative numbers. Thus, when implementing a polynomial $g$, we require that its range match the range of $\theta$, i.e., $\lVert g\rVert_{\infty} \leq \frac{1}{2}$. This can always be satisfied by scaling $g$ accordingly.

As an example, we can add two integers $a$ and $b$, with the conditions $a, b, a + b \in \{-2^{m-1},\dots 0, \dots, 2^{m-1}-1\}$, as follows:
\begin{equation}
    \text{QFT}^{\dagger}_{2^m}F_m\left(\frac{a}{2^m}\right)F_m\left(\frac{b}{2^m}\right)\ket{+}^{\otimes m} = \ket{a + b}.
\end{equation}
Note, the value in the quantum register is really the two's complement of $a + b$.
We define the following controlled operation:
\begin{equation}
\label{eqn:control_op}
    F_{m}(\bm{b}, \theta) := \ketbra{\bm{b}}\otimes F_m(\theta) + (I - \ketbra{\bm{b}})\otimes I,
\end{equation}
where $\bm{b} \in \mathbb{B}^n$. For $S_k \subseteq [n]$, let $\bm{1}_{S_k} \in \mathbb{B}^n$ denote the indicator vector of $S_k$. 
The process for (approximately) loading the value of the polynomial \eqref{eqn:binary_poly} into a quantum register is:
\begin{equation}
\label{eqn:poly_op}
    (I \otimes \text{QFT}^{\dagger}_{2^m})\prod_{k=1}^{K}F_m(\bm{1}_{S_k}, d_k)\ket{\bm{b}}\ket{+}^{\otimes m} = \ket{\bm{b}}\ket{\tilde{g}(\bm{b})},
\end{equation}
where by the assumption on the range of $g$, $|\tilde{g}(\bm{b}) - g(\bm{b})| \leq 2^{-m}$. The result is stored in an auxiliary quantum register of size $O(m)$. The operation $F_m(\bm{b}, \theta)$ requires $m$ $n$-controlled rotation gates. Thus overall it requires $Km$ $n$-controlled rotation gates. An $O(n)$-controlled Toffoli can be implemented with $O(n)$ $\tgate$ gates ~\cite{Jones_2013, maslov2016} and  each controlled rotation can be $\epsilon$-approximately implemented with $O(\log(1/\epsilon))$ $\tgate$'s~\cite{Bocharov_2015, Nam_2020}. Thus, assuming a fixed rotation-gate approximation error the total cost is $O(Kmn)$.

The operation $\text{QFT}_{2^{m}}$ requires $O(m^2)$ gates to be implemented exactly \cite{nielsen2010quantum} and can be implemented approximately, for a fixed approximation error, on a fault-tolerant device with $O(m\log(m))$ $\tgate$ gates \cite{Nam_2020}. For equality constraints, since we will be measuring the entire register containing the value $\tilde{g}(\bm{b})$,  we swap the coherent implementation of the inverse QFT for the semiclassical variant \cite{griffiths1996semiclassical, Parker_2000}. This semiclassical version of the QFT replaces all two-qubit gates with classically-controlled single qubit gates and requires only a single auxiliary qubit that is repeatedly measured and reset to compute the bits of $\tilde{g}(\bm{b})$. Thus, this approach benefits from both mid-circuit measurements and QCL. A fault-tolerant version of this circuit can be approximately implemented with $O(m\log(m))$ $\tgate$ gates \cite{Goto2014ResourceRF}. Thus in a fault-tolerant setting the overall $\tgate$ count of the QFT-based approach is $O(Kmn + m\log(m))$.

\subsection{Initial state construction}
Our proposed approach is flexible with regards to the choice of the initial state, any initial state that is in-constraint suffices. Thus, unlike Ref.~\cite{bartschi2020grover}, when using the complete-graph mixer our approach does not require repeated applications of a unitary and its inverse for preparing the uniform superposition of in-constraint states. However, the initial state we use in experiments discussed in \alt{the \emph{Results} Section}{Sections~\ref{sec:numerical_experiments}~and~\ref{sec:hardware_experiments}} is the uniform superposition over all computational basis states encoding in-constraint solutions. In general, this superposition is hard to prepare. However, there exist constructions for a wide range of practically relevant cases. If the set of feasible solutions is efficiently indexable, Ref.~\cite[Section IIIB]{Marsh_2020} gives an efficient procedure for the initial state preparation. In the specific case of a Hamming-weight equality or inequality constraint, the uniform superposition over feasible states is a superposition of Dicke states with corresponding Hamming weights, which can be constructed efficiently~\cite{bartschi2022short}. Since, our technique does not require the state preparation method be reversible, we can make use of repeat-until-success schemes.

{\subsection{Parameter optimization}}

{The Zeno framework we propose works well with standard techniques used to optimize parameterized quantum circuits. Specifically, as long as each $N_{r}$ is large enough to ensure the desired minimum in-constraint probability is $1-\delta$ (c.f. Corollary~\ref{cor:scheme_2_cor}) for the given parameter range, the direction of steepest descent will still result in a circuit with the same minimum in-constraint probability. Here we make an assumption that $\bm\theta$ remains bounded throughout optimization, which is a valid assumption in practice. This means that both gradient-based and gradient-free local optimization methods can be used with Zeno-augmented hybrid quantum-classical algorithms. A commonly used way to optimize parameterized quantum circuits is to use the parameter-shift rule~\cite{schuld2019evaluating,Wierichs2022generalparameter} in conjunction with a gradient-based optimizer. We now show that the Zeno framework works efficiently with the parameter-shift rule}.

{We consider the task of finding a minimum-eigenvalue state of an observable $M$ using a parameterized quantum evolution consisting of generating Hamiltonians that are also unitary, e.g. L-VQE. We utilize the measurement scheme presented in Equation~\eqref{eqn:ansatz_vqe_with_meas} with the condition that $\forall k, m_k=1$. Following similar arguments as \cite[Section 3]{schuld2019evaluating}, we obtain} 
{\begin{align}
    \frac{\partial}{\partial{\theta_{r}}}\Tr\left\{M\U_Z(\bm{\theta})\rho\U_Z^{\dagger}(\bm{\theta})\right\} 
    &= \sum_{k=1}^{N_{r}}\Tr\left\{M_{k}\mathcal{P}\frac{H_{r}}{N_{r}}e^{-i\frac{\theta_{r}}{N_{r}} H_{r}}\rho_{k} + \text{h.c.}\right\}\nonumber\\
    &=\frac{1}{N_{r}}\sum_{k=1}^{N_{r}}\bigg[\Tr\left\{M\U_Z^{+(r,k)}\rho\U_Z^{\dagger,+(r,k)}\right\} - \Tr\left\{M\U_Z^{-(r,k)}\rho\U_Z^{\dagger,-(r,k)}\right\}\bigg],
\end{align}
where $M_k$ and $\rho_k$ contain terms that have not been differentiated, and  $\U_Z^{\pm(r,k)}$ is the same as $\U_Z(\bm{\theta})$ except that the evolution at the $\sum_{t=1}^{r-1}N_{t} + k$-th step has a phase shift of $\pm\frac{\pi}{4N_r}$. Thus, whereas the normal parameter-shift requires two expectation evaluations per parameter, Zeno would require $2N_{r}$. This is the same additional overhead as in the case of a circuit with gates that share parameters.}

{It also easy to see that the gradient is biased towards minimizing $M_{\mathcal{F}} = P_{\mathcal{F}}MP_{\mathcal{F}}$, i.e. the in-constraint Hamiltonian, as follows:}
{\begin{align}
    \frac{\partial}{\partial{\theta_{r}}}\Tr\left\{M\U_Z(\bm{\theta})\rho\U_Z^{\dagger}(\bm{\theta})\right\} = \mathsf{Pr}_{\mathcal{F}}\frac{\partial}{\partial{\theta_{r}}}\Tr\left\{M_{\mathcal{F}}\frac{\U_Z(\bm{\theta})\rho\U_Z^{\dagger}(\bm{\theta})}{\mathsf{Pr}_{\mathcal{F}}}\right\} + \mathsf{Pr}_{\mathcal{G}}\frac{\partial}{\partial{\theta_{r}}}\Tr\left\{M_{\mathcal{G}}\frac{\U_Z(\bm{\theta})\rho\U_Z^{\dagger}(\bm{\theta})}{\mathsf{Pr}_{\mathcal{G}}}\right\},
\end{align}
where $\mathsf{Pr}_{\mathcal{F}}$ is the probability of projecting onto $\mathcal{F}$ when measuring the parameterized evolution with $\mathcal{P}$. Lastly, Corollary \ref{cor:scheme_2_cor} can be used to ensure $\mathsf{Pr}_{\mathcal{F}} > 1 - \delta$.}

\backmatter

\bmhead{Data Availability}
We make all the data presented in this paper available online at \url{https://doi.org/10.5281/zenodo.7125969}.

\bmhead{Code Availability}
We make the code required to reproduce the figures presented in this paper as well as the code executed on quantum hardware available online at \url{https://doi.org/10.5281/zenodo.7125969}. 

\bmhead{Acknowledgments} The authors wish to thank Antonio Mezzacapo from IBM for his invaluable contributions to this project.  Special thanks also to Tony Uttley, Jenni Strabley and Brian Neyenhuis from Quantinuum for their assistance on the execution of the experiments on the Quantinuum H1-2 trapped-ion quantum processor.

\bmhead{Author Contributions} Marco Pistoia led the overall project. Dylan Herman, Ruslan Shaydulin, Yue Sun and Romina Yalovetzky developed the simulation code and performed numerical experiments. Dylan Herman performed the experiments on trapped-ion quantum processors. Dylan Herman, Ruslan Shaydulin, Yue Sun, Shouvanik Chakrabarti, and Arthur Rattew developed the theoretical results. Shaohan Hu and Piere Minssen contributed to technical discussions. All authors contributed to the writing of the manuscript.

\bmhead{Competing Interests} The authors declare no competing interests.

\bibliography{bibliography}

\section*{Disclaimer}
This paper was prepared for information purposes with contributions from the Global Technology Applied Research center of JPMorgan Chase. This paper is not a product of the Research Department of JPMorgan Chase or its affiliates. Neither JPMorgan Chase nor any of its affiliates make any explicit or implied representation or warranty and none of them accept any liability in connection with this paper, including, but not limited to, the completeness, accuracy, reliability of information contained herein and the potential legal, compliance, tax or accounting effects thereof. This document is not intended as investment research or investment advice, or a recommendation, offer or solicitation for the purchase or sale of any security, financial instrument, financial product or service, or to be used in any way for evaluating the merits of participating in any transaction.

\end{document}